\documentclass[11pt]{article}
\usepackage{amsmath}
\usepackage{amsthm}
\usepackage{amssymb}
\usepackage{latexsym}
\usepackage{graphicx}
\usepackage{tikz}
\usetikzlibrary{patterns}
\numberwithin{equation}{section}
\newtheorem{theo}{Theorem}[section]
\newtheorem{cor}[theo]{Corollary}
\newtheorem{lemma}[theo]{Lemma}
\newtheorem{prop}[theo]{Proposition}
\newtheorem{conj}[theo]{Conjecture}
\theoremstyle{definition}
\newtheorem{defi}[theo]{Definition}
\newtheorem{exa}[theo]{Example}
\newtheorem{rem}[theo]{Remark}

\newcommand{\N}{{\mathbb N}}
\newcommand{\F}{{\mathbb F}}

\newcommand{\Fb}{\overline{\mathbb{F}}}

\newcommand{\cC}{{\mathcal C}}

\newcommand{\cF}{{\mathcal F}}

\newcommand{\cJ}{{\mathcal J}}
\newcommand{\cM}{{\mathcal M}}

\newcommand{\cI}{{\mathcal I}}

\newcommand{\cV}{{\mathcal V}}

\newcommand{\cZ}{{\mathcal Z}}
\newcommand{\medcap}{\mathbin{\scalebox{0.9}{\ensuremath{\bigcap}}}}

\newcommand{\rk}{\mbox{${\rm rk}$\,}}
\newcommand{\im}{\mbox{\rm im}}
\newcommand{\cs}{\mbox{\rm colsp}}
\newcommand{\rs}{\mbox{\rm rowsp}\,}
\newcommand{\Prob}{\mbox{\rm Prob}\,}

\newcommand{\wtH}{\mbox{${\rm wt}_{\rm H}$}}

\newcommand{\T}{\mbox{$\!^{\sf T}$}}
\newcommand{\subspace}[1]{\mbox{$\langle{#1}\rangle$}}
\newcommand{\GL}{\mathrm{GL}}
\newcommand{\supp}{\mathrm{supp}}
\newcommand{\dd}{\textup{d}_{\rm rk}}

\newcommand{\numin}{\mbox{$\nu_{\rm min}$}}

\newcommand{\Smallfourmat}[4]{\mbox{$\left(\begin{smallmatrix}{#1}&{#2}\\{#3}&{#4}\end{smallmatrix}\right)$}}

\newcounter{alp}
\newcounter{ara}
\newcounter{rom}
\newenvironment{romanlist}{\begin{list}{(\roman{rom})\hfill}{\usecounter{rom}
     \topsep0ex \labelwidth.8cm \leftmargin.8cm \labelsep0cm
     \rightmargin0cm \parsep0ex \itemsep.4ex
     \partopsep1ex}}{\end{list}}
\newenvironment{alphalist}{\begin{list}{(\alph{alp})\hfill}{\usecounter{alp}
     \topsep0ex \labelwidth.6cm \leftmargin.6cm \labelsep0cm
     \rightmargin0cm \parsep0ex \itemsep0ex
     \partopsep0ex}}{\end{list}}

\begin{document}
\title{Maximal Ferrers Diagram Codes:\\[.6ex] Constructions and Genericity Considerations}
\date{\today}
\author{Jared Antrobus$^\ast$ and Heide Gluesing-Luerssen\footnote{HGL was partially supported by the grant \#422479 from the Simons Foundation.
  HGL and JA are with the Department of Mathematics, University of Kentucky, Lexington KY 40506-0027, USA;
\{jantrobus,heide.gl\}@uky.edu. Part of the material in this paper was presented at the AMS Sectional Meeting in Columbus, OH, March 2018. }}

\maketitle

{\bf Abstract:}
This paper investigates the construction of rank-metric codes with specified Ferrers diagram shapes.
These codes play a role in the multilevel construction for subspace codes.
A conjecture from 2009 provides an upper bound for the dimension of a rank-metric code with given specified Ferrers diagram shape and rank distance.
While the conjecture in its generality is wide open, several cases have been established in the literature.
This paper contributes further cases of Ferrers diagrams and ranks for which the conjecture holds true.
In addition, the proportion of maximal Ferrers diagram codes within the space of all rank-metric codes with the same shape and dimension is investigated.
Special attention is being paid to MRD codes.
It is shown that for growing field size the limiting proportion depends highly on the Ferrers diagram.
For instance, for $[m\times 2]$-MRD codes with rank~$2$ this limiting proportion is close to $1/e$.

\smallskip


\smallskip


\section{Introduction}
For random linear network coding, see~\cite{CWJ03} by  Chou et al.\ and~\cite{HKMKE03} by Ho et al.,
the natural coding-theoretical objects are subspace codes.
This observation by Koetter et al.~\cite{KoKsch08, SKK08} has led to extensive research efforts for constructions and decoding of subspace codes
\cite{BEGR16,EGRW16,EtSi09,GPB10,GLMT15, GLT16,GoRa14,GoRa17,HeKu17,MaVa13,PNLS17,SiKsch09p,SKK08,TMBR13,WAS13}.

One way to construct good subspace codes utilizes rank-metric codes.
These are subspaces (or subsets) of some matrix space $\F_q^{m\times n}$ endowed with the rank metric $\dd(A,B)\!=\!\rk(A\!-\!B)$.
This naturally leads to the task of constructing large rank-metric codes with a given rank distance, and many of the above
mentioned articles contribute to this question.
Already in the 70's, Delsarte~\cite{Del78} and independently in the 80's Gabidulin~\cite{Gab85} show that the maximum dimension of an $m\times n$-rank-metric code
with rank distance~$\delta$ is given by $m(n-\delta+1)$ if $n\leq m$.
Codes attaining this bound are called MRD codes (maximum rank-distance codes), and both authors provide a construction of such codes.
These MRD codes, now known as Gabidulin codes, are constructed within the $\F_q$-vector space~$\F_{q^m}^n$, which is naturally isometric to $(\F_q^{m\times n},\dd)$.
They are not just $\F_q$-linear but even $\F_{q^m}$-linear.
More recently, a lot of attention has been paid to the existence and construction of MRD codes that are not equivalent to Gabidulin codes and not necessarily $\F_{q^m}$-linear.
Most notably, in \cite{Sh16} Sheekey presents a construction of MRD-codes that are not equivalent to Gabidulin codes.
Further contributions have been made by de la Cruz et al.~\cite{CKWW16}  and Trombetti/Zhou~\cite{TrZh18}.

A very straightforward construction of good subspace codes with the aid of rank-metric codes is the lifting construction~\cite{KoKsch08}:
to each matrix~$M$ in the given rank-metric code one associates the row space of the matrix $(I\mid M)$, where~$I$ is the identity matrix of suitable size.
While this simple construction leads to subspace codes with good distance, it usually does not produce large codes.
A remedy has been introduced by Etzion/Silberstein~\cite{EtSi09}:
obviously a matrix of the form $(I\mid M)\in\F^{m\times(m+n)}$ is in reduced row echelon form (RREF) with pivot indices $1,\ldots,m$.
This observation has led to the  \emph{multilevel construction}, where for each level a rank-metric code in $\F^{m\times n}$ is used to construct a subspace code in
$\F^{m+n}$ with all representing $m\times(m+n)$-matrices being in RREF with a fixed set of general pivot indices.
For this to work out properly, the matrices in the given rank-metric code have to be supported by the Ferrers diagram associated with the list of pivot indices;
see~\cite{EtSi09} and Remark~\ref{R-Multilevel} later in this paper.
As a result, the multilevel construction leads to the task of constructing large Ferrers diagram codes with a given rank distance.
In~\cite{EtSi09} the authors provide an upper bound for the dimension of a rank-metric code supported by a given Ferrers
diagram~$\cF$ and with a given rank distance~$\delta$.
In this paper, codes attaining this bound will be called maximal $[\cF;\delta]$-codes.
To this day, it is not clear whether maximal $[\cF;\delta]$-codes exist for all pairs $(\cF;\delta)$ and all finite fields.
Several cases have been settled by Etzion et al.~\cite{EGRW16,EtSi09} and Gorla/Ravagnani~\cite{GoRa17} and, more recently, by Liu et al.~\cite{LCF19} and Zhang/Ge~\cite{ZhGe19},
but the general case remains widely open.
In~\cite{Bal15} Ballico studies the existence of maximal $[\cF;\delta]$-codes over number fields.

In this paper we survey some of these results and extend them to further classes of pairs $(\cF;\delta)$.
In particular, we provide a family of pairs $(\cF;\delta)$ for which maximal $[\cF;\delta]$-codes can be realized for any finite
field~$\F_q$ as subfield subcodes of Gabidulin codes (or other $\F_{q^m}$-linear MRD codes).
Since Gabidulin codes can be efficiently decoded, the same is true for such subfield subcodes.
We also illustrate that for general pairs $(\cF;\delta)$ such a subfield subcode construction is not possible.
This is due to the non-existence of invariant subspaces in those cases.
Furthermore, we present constructions for the special case where~$\cF$ is the $n\times n$-upper triangle and the rank is~$n-1$.
In this case the dimension of a maximal $[\cF;n-1]$-code is just~$3$, and despite the simplicity of the situation no construction of maximal $[\cF;n-1]$-codes over arbitrary finite fields was known before.

Finally, we turn to the proportion of maximal $[\cF;\delta]$-codes within the space of all $N$-dimensional codes in $\F_q^{m\times n}$ with shape~$\cF$, and where~$N$ is the dimension of a maximal $[\cF;\delta]$-code.
Special attention will be paid to the limiting proportion as~$q$ tends to infinity.
If this limit is~$1$, we call maximal $[\cF;\delta]$-codes generic.
We will see that $[\cF;\delta]$-codes are generic if and only if there exists an $N$-dimensional $[\cF;\delta]$-code over any algebraically closed field of positive characteristic.
This will tell us that genericity depends highly on the shape~$\cF$; in particular
MRD codes are not generic (which has also recently been observed by Byrne/Ravagnani~\cite{ByRa18}).
The latter contrasts recent results in~\cite{NHTRR18} by Neri et al., who showed that MRD codes are generic if one restricts oneself to $\F_{q^m}$-linear rank-metric codes.
Finally, for several nongeneric cases we provide upper bounds on the proportion.
Among other things we will see that the limiting proportion of $[m\times n;\delta]$-MRD codes is upper bounded by
$(1/e)^{(\delta-1)(n-\delta+1)}$ as $q,\,m\rightarrow\infty$, with equality if $n=\delta=2$ (improving upon bounds in~\cite{ByRa18}).
This is derived from the fact~\cite{Sto88} that the proportion of matrices in $\F_q^{m\times m}$ with empty spectrum is asymptotic to $1/e$ as $q,m\rightarrow\infty$.
It remains an open question whether there exist parameters $(m,n,\delta)$ for which the limiting proportion of $[m\times n;\delta]$-MRD codes is zero.

\section{Preliminaries}
Throughout, let $q$ be a prime power and $\F_q$ be a finite field of order~$q$.
For any~$m\in\N$ consider the field extension~$\F_{q^m}$ over~$\F_q$.
Let $B=(x_1,\ldots,x_m)$ be an ordered basis of $\F_{q^m}$ as an $\F_q$-vector space.
Then we have the coordinate map
\[
     \phi_B:\F_{q^m}\longrightarrow\F_q^m,\quad  a:=\sum_{i=1}^m \alpha_ix_i\longmapsto \begin{pmatrix}\alpha_1\\ \vdots\\\alpha_m\end{pmatrix}=:[a]_B
\]
We also write $[a]_B$ for $\phi_B(a)$.
The isomorphism $\phi_B$ extends to the isomorphism
\begin{equation}\label{e-phiB}
  \phi_B: \F_{q^m}^n\longrightarrow\F_q^{m\times n},\ (a_1,\ldots,a_n)\longmapsto \big([a_1]_B,\ldots,[a_n]_B\big).
\end{equation}
On the vector space~$\F_q^{m\times n}$ we define the rank metric as  $\dd(A,B):=\rk(A-B)$. It is well-known that this is indeed a metric.
Furthermore, on the $\F_q$-vector space $\F_{q^m}^n$ we define the rank weight as $\rk(a_1,\ldots,a_n)=\dim_{\F_q}\subspace{a_1,\ldots,a_n}$, where
throughout this paper the notation $\subspace{\ }$ stands for the $\F_q$-subspace generated by the indicated elements.
The rank weight induces the rank metric on $\F_{q^m}^n$ in the obvious way.
It is clear that~$\phi_B$ is an isometry (i.e., a metric-preserving isomorphism) between $\F_{q^m}^n$ and $\F_q^{m\times n}$.

\begin{defi}\label{D-RankMetricCode}
An $\F_q$-subspace of $\F_q^{m\times n}$ or $\F_{q^m}^n$ is called  a \emph{rank-metric code}.
The \emph{(minimal) rank distance} of the rank-metric code~$C$ is defined as $\dd(C):=\min\{\rk(z)\mid z\in C\backslash\{0\}\}$.
An $[m\times n, k; \delta]_q$-code is a rank-metric code in $\F_q^{m\times n}$ or $\F_{q^m}^n$ of $\F_q$-dimension~$k$ and rank distance~$\delta$.
The same terminology will be used for $\F$-subspaces of $\F^{m\times n}$ for an infinite field~$\F$.
\end{defi}

Note that in general a rank-metric code  in $\F_{q^m}^n$  is only required to be $\F_q$-linear and not necessarily $\F_{q^m}$-linear.

A well-studied class of rank-metric codes are those attaining the maximum possible dimension for a given size $m\times n$ and rank distance $\delta$.
In the case where $n\leq m$, the \emph{Singleton bound} tells us that the dimension~$k$ of an $[m\times n,k;\delta]_q$-code is at most $m(n-\delta+1)$, and codes attaining this bound are called \emph{MRD codes} (maximum rank-distance codes), denoted as
$[m\times n;\delta]$-MRD codes.
An MRD code in $\F_{q^m}^n$ may even be an $\F_{q^m}$-linear subspace, in which case we call it an $\F_{q^m}$-linear
$[m\times n;\delta]$-MRD code.

We now turn to matrices supported by Ferrers diagrams.
Throughout, for any $n\in\N$ let $[n]$ denote the set $\{1,\ldots,n\}$.

\begin{defi}\label{D-Ferrers}
A $m\times n$-\emph{Ferrers diagram} $\cF$ is a subset of $[m]\times[n]$ with the following properties:
\begin{romanlist}
\item if $(i,j)\in\cF$ and $j<n$, then $(i,j+1)\in\cF$ (right aligned),
\item if $(i,j)\in\cF$ and $i>1$, then $(i-1,j)\in\cF$ (top aligned).
\end{romanlist}
For $j=1,\ldots,n$ let  $c_j=|\{(i,j)\mid 1\leq i\leq m,\,(i,j)\in\cF\}|$, i.e., $c_j$ is the number of dots in the $j$-th column (see
Figure~\ref{F-F12445}).
We may identify the Ferrers diagram~$\cF$ with the tuple $[c_1,\ldots,c_n]$.
The tuple satisfies $c_1\leq c_2\leq\ldots\leq c_n$.
\end{defi}

Note that we allow $c_1=0$ and $c_n<m$. Thus the size $m\times n$ of $\cF$ is not fixed by the tuple $[c_1,\ldots,c_n]$.
However, for each Ferrers diagram the natural choices of $m$ and $n$ are the number of nonempty rows and columns, respectively. Removing empty rows and columns leads to the case where $c_1>0$ and $c_n=m$.
This is further discussed in the paragraph after Definition \ref{D-Pending}.

\begin{exa}\label{E-F1}
The Ferrers diagram $\cF=[c_1,\ldots,c_n]$ can be visualized as an array of top-aligned and right-aligned dots where the $j$-th column has $c_j$ dots.
Just like for matrices we index the rows from top to bottom and the columns from left to right.
For instance, $\cF=[1,2,4,4,5]$ is given by

     \begin{figure}[ht]
    \centering
     {\small
     \begin{tikzpicture}[scale=0.35]
         \draw (4.5,1.5) node (b1) [label=center:$\bullet$] {};
         \draw (4.5,2.5) node (b1) [label=center:$\bullet$] {};
         \draw (4.5,3.5) node (b1) [label=center:$\bullet$] {};
         \draw (4.5,4.5) node (b1) [label=center:$\bullet$] {};
         \draw (4.5,5.5) node (b1) [label=center:$\bullet$] {};
       \

         \draw (3.5,2.5) node (b1) [label=center:$\bullet$] {};
         \draw (3.5,3.5) node (b1) [label=center:$\bullet$] {};
         \draw (3.5,4.5) node (b1) [label=center:$\bullet$] {};
         \draw (3.5,5.5) node (b1) [label=center:$\bullet$] {};

         \draw (2.5,2.5) node (b1) [label=center:$\bullet$] {};
         \draw (2.5,3.5) node (b1) [label=center:$\bullet$] {};
         \draw (2.5,4.5) node (b1) [label=center:$\bullet$] {};
         \draw (2.5,5.5) node (b1) [label=center:$\bullet$] {};

         \draw (1.5,4.5) node (b1) [label=center:$\bullet$] {};
        \draw (1.5,5.5) node (b1) [label=center:$\bullet$] {};

       \draw (0.5,5.5) node (b1) [label=center:$\bullet$] {};
     \end{tikzpicture}
     }
     \caption{$\cF=[1,2,4,4,5]$}
     \label{F-F12445}
     \end{figure}
\end{exa}

For the rest of this section, let~$\F$ denote an arbitrary, possible infinite field (unless specified otherwise).

\begin{defi}\label{D-Shape}
\begin{alphalist}
\item The support of a matrix $M=(m_{ij})\in\ \F^{m\times n}$ is defined as the set $\supp(M):=\{(i,j)\mid m_{ij}\neq0\}$.
        For a given $m\times n$-Ferrers diagram~$\cF$ we say that $M$ has \emph{shape} $\cF$ if $\supp(M)\subseteq\cF$.
       The subspace of~$\F^{m\times n}$ of all matrices with shape~$\cF$ is denoted by $\F[\cF]$.
\item Let $\cC\subseteq \F^{m\times n}$ be a rank-metric code and let $\cF$ be an $m\times n$-Ferrers diagram.
        If $\cC\subseteq \F[\cF]$, that is, every matrix in~$\cC$ has shape $\cF$, then $\cC$ is called a \emph{Ferrers diagram code} of shape~$\cF$.
        An $[m\times n,k;\delta]$-code in $\F[\cF]$ is called an $[\cF,k;\delta]$-code, or an $[\cF;\delta]$-code if the dimension is not specified.
        If $\F=\F_q$, we also use the notation $[\cF,k;\delta]_q$-code and $[\cF;\delta]_q$-code.
\end{alphalist}
\end{defi}

For the Ferrers diagram $\cF=[m,\ldots,m]$ an $[\cF,k;\delta]_q$-code is thus an $[m\times n,k;\delta]_q$-code.
Note that it does not make sense to talk about $[\cF,k;\delta]_q$-codes in $\F_{q^m}^n$ because the shape of the corresponding matrices in
$\F_q^{m\times n}$ depends on the chosen basis~$B$ for the isomorphism in~\eqref{e-phiB}. We will make use of this fact later in Section~\ref{S-Subspaces}.

\begin{rem}\label{R-Multilevel}
Let us briefly relate Ferrers diagram codes to subspaces codes.
All $m\times n$-matrices with the same Ferrers diagram shape~$\cF$ can be extended to $m\times(m+n)$-matrices in reduced row echelon form (RREF) with the same pivot indices
by inserting standard basis vectors.
For instance, matrices with shape $\cF$ as in Figure~\ref{F-F12445} lead to RREF's of the form
\[
  \begin{pmatrix}1&\bullet&0&\bullet&0&0&\bullet&\bullet&0&\bullet\\
                          0&   0     &1&\bullet&0&0&\bullet&\bullet&0&\bullet\\
                         0&   0     &0&0        &1&0&\bullet&\bullet&0&\bullet\\
                         0&   0     &0&0         &0&1&\bullet&\bullet&0&\bullet\\
                         0&   0     &0&0         &0&0& 0      &0         &1&\bullet \end{pmatrix}.
\]
Precisely, let $\cF=[c_1,\ldots,c_n]$ and set $t_i=|\{j\mid c_j\leq i\}|$ for $i=1,\ldots,m-1$ and $t_0=0$.
Then the pivot indices of the resulting $m\times(m+n)$-matrix in RREF are at positions $t_0+1,t_1+2,t_2+3,\ldots,t_{m-1}+m$.
In this way Ferrers diagram codes give rise to subspace codes via the row spaces of the resulting matrices in RREF.
The multilevel construction by Etzion/Silberstein~\cite{EtSi09}  tells us how to combine various Ferrers shapes to ensure the quality of the subspace code.
Not surprisingly, the rank distance of the Ferrers diagram codes plays a crucial role.
\end{rem}

The above discussion leads to the question as to what the maximum possible dimension~$k$ of an $[\cF,k;\delta]$-code is.
In~\cite{EtSi09} Etzion/Silberstein present an upper bound on the dimension via a simple counting argument.
We need the following notation.

\begin{defi}\label{D-numin}
Let $\cF=[c_1,\ldots,c_n]$ be an $m\times n$-Ferrers diagram and let $\delta\in\N$.
For $j=0,\ldots,\delta-1$ define
\[
  \nu_j:=\nu_j(\cF;\delta)
  =\left\{\begin{array}{l}\text{number of dots in~$\cF$ after removing the}\\ \text{top~$j$ rows and rightmost $\delta-1-j$ columns}\end{array}\right\}
  =\sum_{t=1}^{n-\delta+1+j}\!\!\!\max\{c_t-j,0\}.
\]
Furthermore, set $\numin:=\numin(\cF;\delta)=\min\{\nu_0,\ldots,\nu_{\delta-1}\}$.
\end{defi}

Note that $\numin=0$ whenever $\delta>\min\{m,n\}$.
Moreover, $\numin=0\Longleftrightarrow c_{n-\delta+1+j}\leq j$ for some $j\in\{0,\ldots,\delta-1\}$.
A simple Linear Algebra argument establishes the following upper bound on $[\cF;\delta]_q$-codes.

\begin{theo}[\mbox{\cite[Thm.~1]{EtSi09}}]\label{T-UppB}
Let $\cC\subseteq \F^{m\times n}$ be an $[\cF;\delta]$-code. Then $\dim(\cC)\leq\numin(\cF;\delta)$.
\end{theo}

This gives rise to the following definition.

\begin{defi}\label{D-MaxF}
An $[\cF;\delta]$-code~$\cC\subseteq \F^{m\times n}$ is called \emph{maximal} if $\dim(\cC)=\numin(\cF;\delta)$.
\end{defi}

In the same paper~\cite{EtSi09}, Etzion/Silberstein formulate the following conjecture for Ferrers diagram codes over finite fields.
\begin{conj}\label{C-FConj}
For every $m\times n$-Ferrers diagram~$\cF$, every $1\leq \delta\leq\min\{m,n\}$ and every finite field~$\F_q$ there exists a maximal
$[\cF;\delta]_q$-code.
\end{conj}

The conjecture is certainly true for any~$\cF$ and $\delta=1$: set $\cC=\{E_{ij}\mid (i,j)\in\cF\}$, where $E_{ij}\in\F_q^{m\times n}$ denotes the
standard basis matrix with a one in position $(i,j)$ and zeros elsewhere (this is even true for arbitrary fields).
Conjecture~\ref{C-FConj} has been proven for several cases of~$(\cF;\delta)$ but may still be considered as widely open.
We will revisit some of the established cases later in the paper and settle the conjecture for further cases.
For algebraically closed fields the conjecture is not true in general.
In Section~\ref{S-Prob} we will discuss this more closely and relate the existence of a maximal $[\cF;\delta]$-code over
$\overline{\F_q}$ with genericity over large finite fields.

Let us return to Conjecture~\ref{C-FConj} for finite fields. The simplest case for $\delta\geq2$ is the case where $\cF=[m,\ldots,m]$, that is,~$\cF$ is the full rectangle and does not put any restrictions on the matrices. If without loss of generality $n\leq m$, then $\numin(\cF;\delta)=m(n-\delta+1)$, recovering the Singleton bound.
In other words, a maximal $[\cF;\delta]$-code is an $[m\times n;\delta]$-MRD code.
The existence of such codes has been established by Delsarte~\cite{Del78} and later recovered by Gabidulin~\cite{Gab85}. We recall Gabidulin's construction here, but the two are essentially the same.

\begin{theo}[\mbox{\cite[Thm.~5.4 and 6.3]{Del78}, \cite[Thm.~6/7]{Gab85}}]\label{T-Gab}
Let $m\geq n$ and $g_1,\ldots,g_n\in\F_{q^m}$ be linearly independent over~$\F_q$ and let $\delta\in[n]$. Set $\ell=n-\delta+1$ and
\[
    M:=M(g_1,\ldots,g_n;\ell)=\begin{pmatrix}g_1&\cdots&g_n\\ g_1^q&\cdots &g_n^q\\ \vdots& & \vdots\\ g_1^{q^{\ell-1}}&\cdots&g_n^{q^{\ell-1}}\end{pmatrix}\in\F_{q^m}^{\ell\times n}.
\]
Then the row space $\rs(M):=\{uM\mid u\in\F_{q^m}^{\ell}\}\subseteq\F_{q^m}^n$ is called a \emph{Gabidulin code}.
It is an $\F_{q^m}$-linear $[m\times n;\delta]_q$-MRD code.
\end{theo}

Note that $\cC$ has dimension~$\ell$ over $\F_{q^m}$ (since~$M$ has full row rank), and thus its $\F_q$-dimension is
$m\ell=m(n-\delta+1)$, as desired.

The remainder of this section is devoted to a few simple facts that turn out to be quite useful.
The simple properties given below in Remarks~\ref{R-Fsmaller} and~\ref{R-FDotRem} have already been used in the literature (for instance in the proof of~\cite[Thm.~7]{EGRW16}),
but it seems nonetheless beneficial to formally introduce the according notions.
The terminology in Definition~\ref{D-Pending}(b) below will be particularly convenient.
It is a generalization of~\cite{SiTr15} where the same
notion is used for a more specific case.
The relevance of pending dots is, of course, that if Con\-jec\-ture~\ref{C-FConj} is true, then these dots are not necessary for the existence of
maximal $[\cF;\delta]_q$-codes.

\begin{defi}\label{D-Pending}
\begin{alphalist}
\item Let $\cF_i$ be $m_i\times n$-Ferrers diagrams with the same number of columns. $\cF_1\subseteq\cF_2$ simply means set-theoretic inclusion, thus $(i,j)\in\cF_1$ implies $(i,j)\in\cF_2$.
\item Let $\delta\in[n]$ and $\cF_2$ be an $m_2\times n$-Ferrers diagram.
If there exists an $m_1\times n$-Ferrers diagram $\cF_1\subsetneq\cF_2$ such that $\numin(\cF_1;\delta)=\numin(\cF_2;\delta)$, then the dots in $\cF_2\setminus\cF_1$ are called \emph{pending dots of} $\cF_2$ with respect to~$\delta$.
\end{alphalist}
\end{defi}

Note that comparing two Ferrers diagrams as sets only makes sense in the context where both have the same number of columns.
The diagrams $[1,2,3,4]$ and $[0,1,2,3,4]$ are certainly the same, but as sets of points they look very different.
Recall that Definition \ref{D-Ferrers} includes $m\times n$-Ferrers diagrams~$\cF=[c_1,\ldots,c_n]$ with $c_1=0$ and $c_n<m$.
This allows us to pad a diagram with empty rows and columns to make it a desired size for the purpose of comparison.
In the same way we may delete empty rows or columns in order to obtain a Ferrers diagram where the first column and last row are
non-empty.

\begin{rem}\label{R-Fsmaller}
Let $\cF_1\subseteq\cF_2$  be  $m_i\times n$-Ferrers diagrams such that $\numin(\cF_1;\delta)=\numin(\cF_2;\delta)$.
Then the  existence of a maximal $[\cF_1;\delta]$-code implies the existence of a maximal $[\cF_2;\delta]$-code over the same field.
This is clear because each matrix with shape~$\cF_1$ also has shape~$\cF_2$.
\end{rem}

\begin{exa}\label{E-F11}
\begin{alphalist}
\item Consider the Ferrers diagram $\cF=[1,2,4,4,5]$ shown in Figure~\ref{F-F12445}.
         Then~$\cF$ does not have any pending dots with respect to $\delta=2$ or $\delta=3$, but the dot at position $(4,3)$ is pending with respect to~$\delta=4$.
\item For $\cF=[1,3,3,4,5]$ and $\delta=4$ the dots at positions $(1,1)$ and $(2,3)$ are both pending as individual dots, that is, removing either one of them does not
       decrease $\numin(\cF;\delta)=4$. However, removing both of them will decrease it to~$3$.
\item In \cite[Ex.~5]{EGRW16} the authors present a construction for a maximal $[\cF;3]_q$-code where
        $\cF=[2,4,4,6,8]$ for fields $\F_q$ with $q\geq4$.
        However, the bottom~$4$ dots are pending and removing them leads to the Ferrers diagram $\cF'=[2,4,4,5,5]$, for which the authors present
        a construction of maximal $[\cF';3]_q$-codes for arbitrary fields in \cite[Thm.~2]{EGRW16}.
        Thus, not only does the latter construction work for all finite fields, it also does not need the positions of the pending dots.
        We will revisit  \cite[Thm.~2]{EGRW16} in Theorem~\ref{T-Fn1}.
\end{alphalist}
\end{exa}

\begin{rem}\label{R-FDotRem}
Let $\delta\in[n]$ and $\cF',\,\cF$ be $m\times n$-Ferrers diagrams such that $\cF'\subsetneq\cF$ and $|\cF\,\setminus\,\cF'|=1$ (that is,~$\cF'$ is
obtained from~$\cF$ by removing one dot).
Suppose $\numin(\cF';\delta)=\numin(\cF;\delta)-1$.
Then the existence of a maximal $[\cF;\delta]$-code over the field~$\F$ implies the existence of a maximal $[\cF';\delta]$-code over~$\F$.
Indeed, let $\cC$ be an $[\cF,k;\delta]$-code, where $k=\numin(\cF;\delta)$.
Let $\{(i,j)\}=\cF\,\setminus\cF'$.
We can clearly choose a basis $\{A_1,\ldots,A_k\}$ of~$\cC$ such that $(A_s)_{i,j}=0$ for all $1\leq s\leq k-1$.
Then $\{A_1,\ldots,A_{k-1}\}$ is a basis of a maximal $[\cF';\delta]$-code.
\end{rem}

This reduction technique is certainly not a new result, and is fairly obvious. Nevertheless, we include the following simple example to illustrate its power.

\begin{exa}\label{E-4Flow}
Let~$\delta=3$.
Figure~\ref{F-44d3} shows all $4\times4$-Ferrers diagrams for which Conjecture~\ref{C-FConj} can be confirmed
via reduction described in Remark~\ref{R-FDotRem} starting with a $[4\times 4;3]$-MRD code.
The number in the bottom right corner is $\numin(\cF;3)$ for the given Ferrers diagram~$\cF$.
Later in this paper we will establish Conjecture~\ref{C-FConj} for $n\times n$-upper triangular matrices with $\delta=3$ (see Corollary~\ref{C-Upper}).
Figure~\ref{F-44d3a} shows all $4\times4$-Ferrers diagrams for which Conjecture~\ref{C-FConj} can be confirmed
via reduction as in Remark~\ref{R-FDotRem} starting from the upper triangular shape.
\end{exa}

\begin{figure}[ht]
\centering
\includegraphics[height=7.7cm]{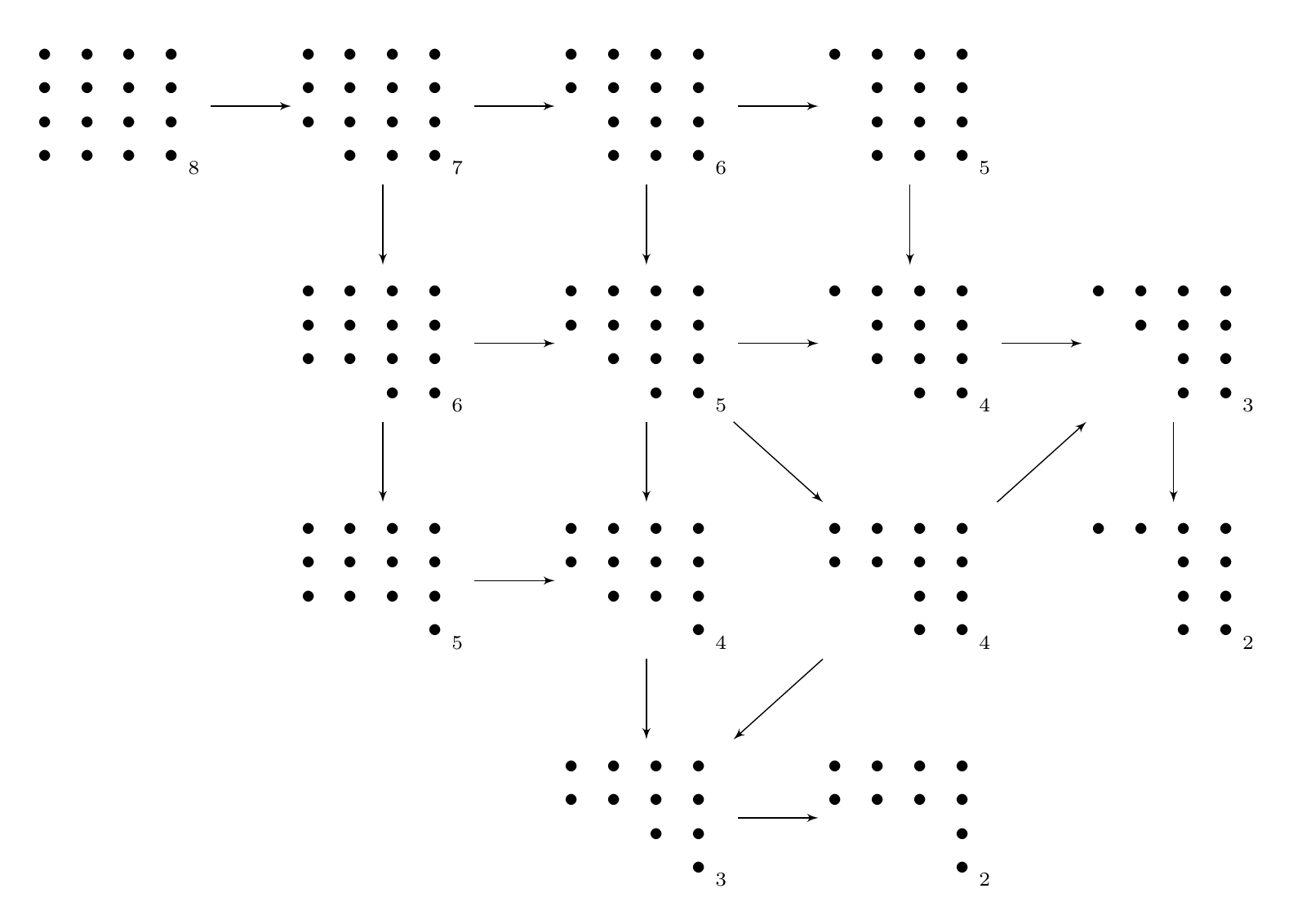}
\caption{Reduction for $4\times 4$-diagrams with $\delta=3$ starting from $\cF=[4,4,4,4]$.}
\label{F-44d3}
\end{figure}

\begin{figure}[ht]
\centering
\includegraphics[height=3.8cm]{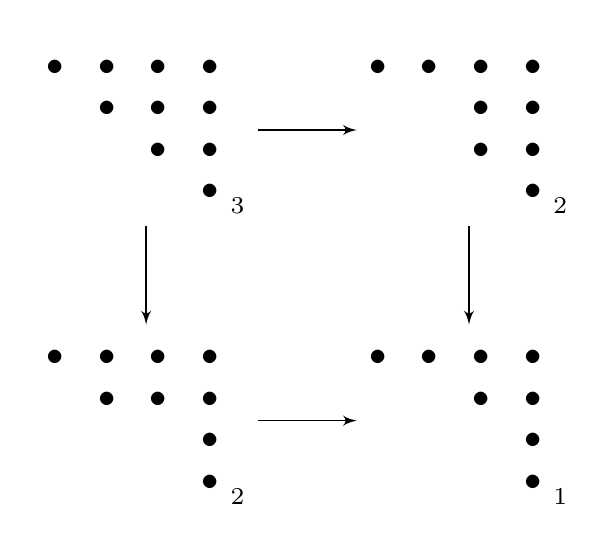}
\caption{Reduction for $4\times 4$-diagrams with $\delta=3$ starting from $\cF=[1,2,3,4]$.}
\label{F-44d3a}
\end{figure}

The only $4\times4$-Ferrers diagram with positive $\numin$ not appearing in these charts is $\cF=[1,3,3,4]$.
This case has been dealt with by Etzion et al.~\cite[Ex.~7]{EGRW16} by making use of a suitable extension of a Gabidulin code.
We present a simple alternative construction.

\begin{exa}\label{E-F1334}
Let $\delta=3$ and consider the $4\times4$-Ferrers diagram $\cF=[1,3,3,4]$ shown in Figure~\ref{F-F1334}.
Then $\numin=4$.
In order to construct a maximal $[\cF;3]_q$-code over any finite field $\F=\F_q$, we start with a $[3\times3;3]$-MRD code,  hence its
dimension is~$3$.

     \begin{figure}[ht]
     \centering
     {\small
     \begin{tikzpicture}[scale=0.35]

         \draw (3.5,2.5) node (b1) [label=center:$\bullet$] {};
         \draw (3.5,3.5) node (b1) [label=center:$\bullet$] {};
         \draw (3.5,4.5) node (b1) [label=center:$\bullet$] {};
         \draw (3.5,5.5) node (b1) [label=center:$\bullet$] {};

         \draw (2.5,3.5) node (b1) [label=center:$\bullet$] {};
         \draw (2.5,4.5) node (b1) [label=center:$\bullet$] {};
         \draw (2.5,5.5) node (b1) [label=center:$\bullet$] {};

         \draw (1.5,3.5) node (b1) [label=center:$\bullet$] {};
         \draw (1.5,4.5) node (b1) [label=center:$\bullet$] {};
        \draw (1.5,5.5) node (b1) [label=center:$\bullet$] {};

       \draw (0.5,5.5) node (b1) [label=center:$\bullet$] {};
     \end{tikzpicture}
     }
     \caption{$\cF=[1,3,3,4]$}
     \label{F-F1334}
     \end{figure}

\noindent We may choose a basis $B_1,B_2,B_3$ of this code in the form
\[
   B_1=\begin{pmatrix}1&a_{12}^{(1)}&a_{13}^{(1)}\\0&a_{22}^{(1)}&a_{23}^{(1)}\\0&a_{32}^{(1)}&a_{33}^{(1)}\end{pmatrix},\
   B_2=\begin{pmatrix}0&a_{12}^{(2)}&a_{13}^{(2)}\\1&a_{22}^{(2)}&a_{23}^{(2)}\\0&a_{32}^{(2)}&a_{33}^{(2)}\end{pmatrix},\
   B_3=\begin{pmatrix}0&a_{12}^{(3)}&a_{13}^{(3)}\\0&a_{22}^{(3)}&a_{23}^{(3)}\\1&a_{32}^{(3)}&a_{33}^{(3)}\end{pmatrix}
\]
(see also Example~\ref{E-CompMatMRD}(a) below).
Hence the general linear combination is
\[
  B(\lambda):=\lambda_1 B_1+\lambda_2B_2+\lambda_3B_3=\begin{pmatrix}\lambda_1&p_{12}&p_{13}\\\lambda_2&p_{22}&p_{23}\\\lambda_3&p_{32}&p_{33}\end{pmatrix},
  \text{ where } p_{ij}=\sum_{\ell=1}^3 a_{ij}^{(\ell)}\lambda_\ell.
\]
Rank distance~$3$ guarantees that $(a_{22}^{(1)},a_{32}^{(1)})\neq(0,0)$. We assume without loss of generality that $a_{22}^{(1)}\neq0$.
Define now $A_1,\ldots,A_4\in\F^{4\times 4}$ such that their general linear combination has the form
\[
  A(\lambda)=\sum_{\ell=1}^4\lambda_\ell A_\ell=\begin{pmatrix}\lambda_4&\lambda_1&p_{12}&p_{13}\\0&\lambda_2&p_{22}&p_{23}\\0&\lambda_3&p_{32}+\lambda_4&p_{33}\\0&0&0&\lambda_4\end{pmatrix}.
\]
It remains to show that $\rk(A(\lambda))\geq 3$ for all $\lambda=(\lambda_1,\ldots,\lambda_4)\neq0$.
This is clear if $\lambda_4=0$.
Thus let $\lambda_4\neq0$. In this case
\[
  \rk A(\lambda)\geq3\Longleftrightarrow
  \rk\begin{pmatrix}\lambda_2&p_{22}\\ \lambda_3&p_{32}+\lambda_4\end{pmatrix}\geq1.
\]
The right hand side is clearly true if $(\lambda_2,\lambda_3)\neq(0,0)$.
In the case where $(\lambda_2,\lambda_3)=(0,0)$, the matrix on the right hand side has
second column $(a_{22}^{(1)}\lambda_1,\,a_{32}^{(1)}\lambda_1+\lambda_4)\T$, and the
assumption $a_{22}^{(1)}\neq0$ along with $(\lambda_1,\lambda_4)\neq(0,0)$ guarantees that this vector is nonzero.
All of this establishes the existence of optimal $[\cF;3]_q$-codes over any field~$\F_q$.
We will return to this particular Ferrers diagram~$\cF$ in Example~\ref{E-F1334a} and Corollary~\ref{C-Prob1334}/Example~\ref{E-F1334b}.
In the former we will show that a maximal $[\cF;3]_q$-code cannot be found as an $\F_q$-linear subspace of an
$\F_{q^4}$-linear $[4\times 4;3]$-MRD code.
In the latter we will discuss the probability that a random choice of~$4$ matrices with shape~$\cF$ generate a maximal $[\cF;3]_q$-code.
\end{exa}

We close the section with a well-known example utilizing companion matrices.
Part~(b) and~(c) below are simple instances of the aforementioned reduction methods.

\begin{exa}\label{E-CompMatMRD}
\begin{alphalist}
\item Consider the case $m=n=\delta$, thus $\ell=1$. Let $B=(1,\alpha,\ldots,\alpha^{m-1})$ be a basis of $\F_{q^m}$ over~$\F_q$, and consider
        the matrix $M=(1,\alpha,\ldots,\alpha^{m-1})\in\F_{q^m}^{1\times m}$.
        Let $f=\sum_{i=0}^m f_i x^i\in\F_q[x]$ be the monic minimal polynomial of~$\alpha$ over~$\F_q$ (thus $f_m=1$).
        Then the matrix code $\phi_B(\rs(M))\subseteq\F_q^{m\times m}$ is the $m$-dimensional code given by

        \[
            \phi_B(\rs(M))=\subspace{I,\,C,\ldots,C^{m-1}}, \text{ where }
            C=\begin{pmatrix} 0&0&\cdots& 0&-f_0\\ 1&0&\cdots&0&-f_1\\0&1&\cdots&0&-f_2\\ & &\ddots  & &\vdots\\ 0&0&\cdots&1&-f_{m-1}\end{pmatrix},
        \]
        that is,~$C$ is the companion matrix of~$f$.
        For any $i\in[m]$ the code $\cC=\subspace{I,\,C,\ldots,C^{i-1}}$ is a maximal $[\cF;m]_q$-code for the $m\times m$-Ferrers diagram
        $\cF=[i,i+1,\ldots,m,\ldots,m]$
        (thus $c_t=\min\{i-1+t,m\}$ and the last $i$ columns have~$m$ dots).
\item The previous code can be used to cover further Ferrers diagrams.
        Choose $t\leq i-1$ and delete the~$t$ rightmost columns of all matrices in~$\cC$.
        This yields an $m\times n$-Ferrers diagram code~$\tilde{\cC}$ with shape $\tilde{\cF}=[i,i+1,\ldots,m,\ldots,m]$, where $n=m-t$
        (and the rightmost~$i-t$ columns have~$m$ dots).
        The code~$\tilde{\cC}$ clearly has dimension~$i$ and thus is a  maximal $[\tilde{\cF};n]_q$-code
        because $\numin(\tilde{\cF};n)\leq\nu_0(\tilde{\cF};n)=i$.
\item We can go even further. Consider an $m\times n$-Ferrers diagram $\cF=[c_1,\ldots,c_n]$ where $c_j\geq c_1+j-1$ for $j=2,\ldots,n$ (hence $c_1\leq m-n+1$).
        Then $\numin(\cF;n)=c_1$ and this remains true even after removing the dots at positions $(i,j)$ with $i>c_1+j-1$, i.e., these dots are pending.
        Removing them leads to the Ferrers diagram $\tilde{\cF}$ as in~(b) with $i=c_1$.
       Hence there exists a maximal $[\cF;n]_q$-code.
\end{alphalist}
\end{exa}

\section{Maximal Ferrers Diagram Codes as Subspaces of MRD Codes}\label{S-Subspaces}
In this section we present a class of pairs $(\cF;\delta)$ for which maximal $[\cF;\delta]_q$-codes can be found as
$\F_q$-subspaces of some (in fact any) $\F_{q^m}$-linear MRD code with the same rank distance.

We start with two well-known results (Theorem \ref{T-Fn1} and Corollary \ref{C-cn}) and their proofs, which will help to generalize them.
For the rest of the paper we fix $n\leq m$, and throughout this section we assume $2\leq\delta\leq n$ (as the existence of maximal $[\cF;1]$-codes is trivial).

Recall the isomorphism $\phi_B:\F_{q^m}^n\longrightarrow\F_q^{m\times n}$ from~\eqref{e-phiB} based on a chosen ordered basis~$B$ of~$\F_{q^m}$ over~$\F_q$.
For the following result note that every $\F_{q^m}$-linear MRD code in $\F_{q^m}^n$ has a systematic generator matrix. This is a consequence of \cite[Thm.~2]{Gab85}.

\begin{theo}[\mbox{\cite[Thm.~2]{EtSi09}, \cite[Sec.~2.5]{GaPi13}, \cite[Cor.~19]{GoRa17}}]\label{T-Fn1}
Let $\cF=[c_1,\ldots,c_n]$ be an $m\times n$-Ferrers diagram such that $c_j=m$ for all $j=n-\delta+2,\ldots,n$
(that is, the last $\delta-1$ columns of~$\cF$ have the maximum number of $m$ dots).
Set $\ell=n-\delta+1$ and let $G=(I_\ell\mid A)\in\F_{q^m}^{\ell\times n}$ be a generator matrix of an $\F_{q^m}$-linear $[m\times n;\delta]$-MRD code (for instance, a Gabidulin code).
Let $B=(x_1,\ldots,x_m)$ be an ordered basis of $\F_{q^m}$ over~$\F_q$.
Then the subspace
\[
   \cC=\left\lbrace\phi_B\big((u_1,\ldots,u_\ell)G\big)\mid u_t\in\subspace{x_1,\ldots,x_{c_t}}\text{ for }t=1,\ldots,\ell\right\rbrace\subseteq\F_q^{m\times n}
\]
is a maximal $[\cF;\delta]_q$-code.
Furthermore, $\numin(\cF;\delta)=\nu_0=\sum_{t=1}^\ell c_t$.
\end{theo}

\begin{proof}
Note that~$\cC$ is clearly an $\F_q$-vector space.
Next, let $(u_1,\ldots,u_\ell)\in\F_{q^m}^\ell$ be such that $u_t\in V_t:=\subspace{x_1,\ldots,x_{c_t}}$. Set $(u_1,\ldots,u_\ell)A=(v_1,\ldots,v_{n-\ell})$. Then
\[
  \phi_B\big((u_1,\ldots,u_\ell)G\big)=\big([u_1]_B,\ldots,[u_\ell]_B,[v_1]_B,\ldots,[v_{n-\ell}]_B\big)=:M.
\]
By choice of $u_t$, it follows that the matrix~$M$ has indeed shape~$\cF$.
Here it is crucial that the last $\delta-1$ columns of~$\cF$ are full and therefore do not impose any restrictions on the coordinate vectors of~$v_1,\ldots,v_{n-\ell}$.
Clearly, $\dd(\cC)=:\delta'\geq\delta$ because~$\cC$ is a subspace of an MRD code of distance~$\delta$.
Finally, $\dim_{\F_q}(\cC)=\sum_{t=1}^\ell \dim V_t=\sum_{t=1}^\ell c_t=\nu_0(\cF;\delta)\geq\nu_0(\cF;\delta')\geq\numin(\cF;\delta')$,
where the first inequality is strict iff $\delta'>\delta$.
Now the upper bound in Theorem~\ref{T-UppB} implies $\delta'=\delta$ and $\dim_{\F_q}(\cC)=\nu_0(\cF;\delta)=\numin(\cF;\delta)$.
\end{proof}

One may note that, once $\numin(\cF;\delta)=\nu_0=\sum_{t=1}^\ell c_t$ is established, the result above also follows from the reduction process
described in Remark~\ref{R-FDotRem}.
Indeed, for $\hat{\cF}=[m]\times[n]$ we have $\numin(\hat{\cF};\delta)=m(n-\delta+1)$ and, since $c_t=m$ for $t>\ell$, we conclude
$\numin(\cF;\delta)=\numin(\hat{\cF};\delta)-|\hat{\cF}\setminus\cF|$.
This has already been observed in~\cite[Rem.~6]{EtSi09} and \cite[Cor.~19]{GoRa17}.

Since we may always reduce to the case where $c_n=m$ by removing empty rows, the next result follows immediately.
\begin{cor}\label{C-Delta2}
Let $\delta=2$. Then Conjecture~\ref{C-FConj} holds true for all Ferrers diagrams~$\cF$ and fields~$\F_q$.
\end{cor}

The next result bears similarity to Theorem \ref{T-Fn1}, but arrives at the same conclusion with a weaker assumption thanks to the consideration of pending dots.

\begin{cor}[\mbox{\cite[Thm.~3]{EGRW16} and \cite[Thm.~23]{GoRa17}}]\label{C-cn}
Let $\cF=[c_1,\ldots,c_n]$ be an $m\times n$-Ferrers diagram such that $c_j\geq n$ for all $j=n-\delta+2,\ldots,n$ (that is, the last $\delta-1$ columns
have at least~$n$ dots).
Then there exists a maximal $[\cF;\delta]_q$-code.
More precisely, all dots at positions $(i,j)$ with $i>\hat{m}=\max\{c_{n-\delta+1},n\}$ are pending, and there exists a maximal $[\hat{\cF};\delta]_q$-code where $\hat{\cF}=[\hat{c}_1,\ldots,\hat{c}_n]$ with
$\hat{c}_t=\min\{c_t,\hat{m}\}$.
\end{cor}

\begin{proof}
Set $\ell=n-\delta+1$.
We show first that $\numin(\cF;\delta)=\nu_0=\sum_{t=1}^{\ell} c_t$.
Using Definition~\ref{D-numin} we compute for any $j=1,\ldots,\delta-1$
\[
  \nu_j=\sum_{t=1}^{\ell+j}\max\{c_t-j,0\}\geq \sum_{t=1}^{\ell}(c_t-j)+\sum_{t=\ell+1}^{\ell+j} (n-j)
          \geq \nu_0+j(n-j-\ell)\geq\nu_0.
\]
Let now $\hat{m}=\max\{c_\ell,n\}$ and consider the $\hat{m}\times n$-Ferrers diagram $\hat{\cF}=[\hat{c}_1,\ldots,\hat{c}_n]$, where
\[
   \hat{c}_t=\min\{c_t,\hat{m}\}=\left\{\begin{array}{ll} c_t,&\text{for }t=1,\ldots,\ell,\\ \hat{m},&\text{for }t=\ell+1,\ldots,n.\end{array}\right.
\]
Then the Ferrers diagram~$\hat{\cF}$ satisfies the assumptions of Theorem~\ref{T-Fn1}.
Thus there exists a maximal $[\hat{\cF};\delta]_q$-code and its dimension is given by $\numin(\hat{\cF};\delta)=\sum_{t=1}^\ell \hat{c}_t=\sum_{t=1}^\ell c_t=\numin(\cF;\delta)$.
Since  $\hat{\cF}\subseteq\cF$ Remark~\ref{R-Fsmaller}  concludes the proof.
\end{proof}

In \cite[Thm. 8]{EGRW16}, Etzion et al.\ take the above idea further, tackling the case where the rightmost $\delta-1$ columns contain
at least $n-1$ dots, assuming other criteria were also met.
More recently Liu et al.~\cite[Thm.~3.13]{LCF19} generalize the argument to handle $n-r$ dots, again requiring further restrictions on the shape.
In particular, the first~$r$ columns combined may have no more than $m-n+r$ dots.

In Theorem~\ref{T-Fn1}  we could choose any ordered basis~$B$ to obtain the desired maximal $[\cF;\delta]_q$-code as a subfield subcode.
In Theorem~\ref{T-Fd2} below we will prove a generalization of Theorem~\ref{T-Fn1} for which we will have to make a judicious choice of basis.
The construction and assumptions differ from \cite[Thm.~3.13]{LCF19}.
We first need the following lemmas.

\begin{lemma}\label{L-InterSub}
Let $V$ be an $m$-dimensional vector space and $V_1,\ldots,V_t$ be subspaces of $V$ with $\dim V_j\geq d_j$.
Then $\dim\big(\medcap_{j=1}^t V_j\big)\geq \sum_{j=1}^t d_j -(t-1)m$.
\end{lemma}

\begin{proof}
We induct on the number of subspaces. Clearly the statement holds for $t=1$. Assume the statement holds for $t-1$ subspaces. Then
\begin{align*}
\dim\big(\medcap_{j=1}^t V_j\big)
  &=\dim\Big(V_t\cap\medcap_{j=1}^{t-1} V_j\Big)=\dim V_t+\dim\big(\medcap_{j=1}^{t-1} V_j\big)
     -\underbrace{\dim\big(V_t+\medcap_{j=1}^{t-1} V_j\big)}_{\leq m}\\
   &\geq d_t+\sum_{j=1}^{t-1}d_j-(t-2)m-m=\sum_{j=1}^t d_j -(t-1)m.
   \qedhere
\end{align*}
\end{proof}

\begin{lemma}\label{L-MRDLinInd}
Let $G=(I_\ell\mid A)\in\F_{q^m}^{\ell\times n}$ be the generator matrix of an $\F_{q^m}$-linear MRD code (thus, its rank distance is $n-\ell+1$).
Let $A=(a_{ij})$.
Then $\rk(1,a_{1j},\ldots,a_{\ell j})=\ell+1$ for all $j=1,\ldots,n-\ell$, i.e.,  the entries of this vector are linearly independent over~$\F_q$.
In particular, $a_{ij}\not\in\F_q$ for all $(i,j)$.
\end{lemma}

\begin{proof}
Consider without loss of generality $j=1$.
Let $\lambda_0+\sum_{i=1}^\ell\lambda_i a_{i1}=0$ for some $\lambda_i\in\F_q$.
Then $(\lambda_1,\ldots,\lambda_\ell)G=(\lambda_1,\ldots,\lambda_\ell,-\lambda_0,b_1,\ldots,b_{n-\ell-1})$ for some $b_i\in\F_{q^m}$.
Since all $\lambda_i$ are in $\F_q$, this vector has rank at most $n-\ell$, whereas the code has distance $n-\ell+1$.
Thus the vector is zero and hence $\lambda_i=0$ for all~$i$, as desired.
\end{proof}

Now we are ready to establish the following result.

\begin{theo}\label{T-Fd2}
Let $\cF=[c_1,\ldots,c_n]$ be an $m\times n$-Ferrers diagram. Let $2\leq\delta\leq n$ and $\ell=n-\delta+1$.
Set $\varepsilon=\sum_{t=\ell+1}^n(m-c_t)$, that is, $\varepsilon$ is the number of dots missing in the rightmost $\delta-1$ columns of~$\cF$.
Suppose
\begin{equation}\label{e-ct}
  c_t\leq c_{\ell+1}-\varepsilon(\ell+1-t) \text{ for } t=1,\ldots,\ell.
\end{equation}
Let $G=(I_\ell\mid A)\in\F_{q^m}^{\ell\times n}$ be the generator matrix of an $\F_{q^m}$-linear $[m\times n;\delta]$-MRD code.
Then there exists an ordered basis $B=(x_1,\ldots,x_m)$ of $\F_{q^m}$ over~$\F_q$ such that
the subspace
\begin{equation}\label{e-Cd2}
   \cC=\left\lbrace\phi_B\big((u_1,\ldots,u_\ell)G\big)\mid u_t\in\subspace{x_1,\ldots,x_{c_t}}\text{ for }t=1,\ldots,\ell\right\rbrace
\end{equation}
is a maximal $[\cF;\delta]_q$-code.
In this case $\numin(\cF;\delta)=\nu_0=\sum_{t=1}^\ell c_t$.
\end{theo}

Theorem~\ref{T-Fn1} is the special case where $\varepsilon=0$. In this case the inequalities~\eqref{e-ct} are vacuous.

\begin{proof}
For any $u=(u_1,\ldots,u_\ell)\in\F_{q^m}^\ell$ we have
\begin{equation}\label{e-uv}
  uG=(u_1,\ldots,u_\ell,v_1,\ldots,v_{n-\ell}),\text{ where } (v_1,\ldots,v_{n-\ell})=uA.
\end{equation}
As in the proof of Theorem~\ref{T-Fn1}, for any fixed basis~$B$,  we may choose $u_1,\ldots,u_\ell$ such that the first~$\ell$ columns of the matrix $\phi_B(uG)$ adhere to the desired shape~$\cF$.
However, now we also have to accommodate the last~$n-\ell$ columns.
We show that for a specific choice of basis~$B$ this can indeed be achieved.

Let $A=(a_{ij})_{i\in[\ell]}^{j\in[n-\ell]}$. Then $a_{ij}\not\in\F_q$ for all $i,j$ thanks to Lemma~\ref{L-MRDLinInd}. In particular, $a_{ij}\neq0$.
Consider any chain of subspaces
\[
   V_1\subsetneq V_2\subsetneq\ldots\subsetneq V_m=\F_{q^m},
\]
such that $\dim V_i=i$.
For $t\in[\ell]$ set $W_t=\bigcap_{j=1}^{n-\ell} V_{c_{\ell+j}}a_{tj}^{-1}$.
Since $\dim(V_{c_{\ell+j}}a_{tj}^{-1})=c_{\ell+j}$, Lemma~\ref{L-InterSub} implies that
\[
      \dim (W_t)\geq \sum_{j=1}^{n-\ell} c_{\ell+j}-(n-\ell-1)m= m-\varepsilon\ \text{ for all }t\in[\ell].
\]
Consider the chain of subspaces
\[
  V_{c_{\ell+1}}\cap\bigcap_{j=1}^\ell W_j\subseteq V_{c_{\ell+1}}\cap\bigcap_{j=2}^\ell W_j\subseteq\ldots\subseteq V_{c_{\ell+1}}\cap W_\ell\subseteq V_{c_{\ell+1}}\subseteq\F_{q^m}.
\]
By Lemma~\ref{L-InterSub} and~\eqref{e-ct} we have for $t\in[\ell]$
\begin{align*}
   \dim\bigg(V_{c_{\ell+1}}\cap\bigcap_{j=t}^\ell W_j\bigg)&\geq  c_{\ell+1}+\sum_{j=t}^\ell \dim(W_j)-(\ell-t+1)m\\
      &\geq c_{\ell+1}+(\ell-t+1)(m-\varepsilon)-(\ell-t+1)m=c_{\ell+1}-(\ell-t+1)\varepsilon\\
      &\geq c_t.
\end{align*}
This allows us to choose an ordered basis $B=(x_1,\ldots,x_m)$ of $\F_{q^m}$ such that
\[
  x_1,\ldots,x_{c_t}\in V_{c_{\ell+1}}\cap\bigcap_{j=t}^\ell W_j\ \text{ for }t\in[\ell].
\]
Now we can prove that the code~$\cC$ in~\eqref{e-Cd2} has shape~$\cF$.
Consider $uG$ as in~\eqref{e-uv}, and where $u_t\in\subspace{x_1,\ldots,x_{c_t}}$.
Then the first~$\ell$ columns of $\phi_B(uG)$ conform to the shape of~$\cF$.
Moreover,
\[
  u_t a_{tj}\in\subspace{x_1,\ldots,x_{c_t}}a_{tj}\subseteq W_t a_{tj}\subseteq V_{c_{\ell+j}}\text{ for }t\in[\ell]\text{ and }j\in[n-\ell].
\]
Thus $v_j=\sum_{t=1}^\ell u_ta_{tj}\in V_{c_{\ell+j}}$ for $j\in[n-\ell]$ and all of this shows that $\phi_B(uG)$ indeed has shape~$\cF$.
Finally, $\sum_{t=1}^\ell\dim\subspace{x_1,\ldots,x_{c_t}}=\sum_{t=1}^\ell c_t=\nu_0(\cF;\delta)\geq\numin(\cF;\delta')$, where
$\delta'\geq\delta$ is the rank distance of~$\cC$.
As in the proof of Theorem~\ref{T-Fn1} we conclude that $\delta'=\delta$ and the code~$\cC$ in~\eqref{e-Cd2} is a maximal $[\cF;\delta]_q$-code.
\end{proof}

The inequalities~\eqref{e-ct} can be regarded as a staircase condition: the first~$\ell$ columns must not have any dots below the staircase which starts at the last dot
in column~$\ell+1$ and goes left and upward with step size~$\varepsilon$; see the next example.
In fact, Inequality~\eqref{e-ct} is trivially true for $t>\ell$ and thus \emph{no} column reaches below the staircase.

We wish to point out that in~\cite[Thm.~3.2 and~3.6]{ZhGe19} Zhang/Ge establish the existence of further cases of maximal
$[\cF;\delta]_q$-codes by imposing a rapid increase of the column indices.
The conditions are very different from ours and imply the existence of a tower of subfields of~$\F_{q^m}$.
As the examples in~\cite{ZhGe19} show, in most cases a large number of pending dots is used for the constructions.

\begin{exa}\label{E-StaircaseGeneral}
Consider $\cF=[1,3,5,7,7,8,8,8]$ and $\delta=6$. Then $\ell=3$ and $\varepsilon=2$.
       The staircase condition~\eqref{e-ct} is indeed satisfied as can also be seen by Figure~\ref{F-Big}.
       Hence maximal $[\cF;6]_q$-codes exist over every field~$\F_q$.
       Note that the three dots in the bottom row are pending in the sense of Definition~\ref{D-Pending}.
      However, deleting them leads to a Ferrers diagram with fewer rows than columns.
     Swapping rows and columns accordingly yields the $8\times7$- Ferrers diagram $\tilde{\cF}=[5,5,6,6,7,7,8]$.
     No previous construction provides us with a maximal $[\tilde{\cF};6]_q$-code and thus Remark~\ref{R-Fsmaller} cannot be utilized for the given pair $(\cF;6)$.
      \begin{figure}[ht]
     \centering
     {\small
     \begin{tikzpicture}[scale=0.35]
         \draw (3,1.2) to (3,3);
         \draw (3,3) to (2,3);
         \draw (2,3) to (2,5);
         \draw (2,5) to (1,5);
        \draw (1,5) to (1,7);
        \draw (1,7) to (0,7);
            \draw (7.5,.5) node (b1) [label=center:$\bullet$] {};
         \draw (7.5,1.5) node (b1) [label=center:$\bullet$] {};
         \draw (7.5,2.5) node (b1) [label=center:$\bullet$] {};
         \draw (7.5,3.5) node (b1) [label=center:$\bullet$] {};
         \draw (7.5,4.5) node (b1) [label=center:$\bullet$] {};
         \draw (7.5,5.5) node (b1) [label=center:$\bullet$] {};
         \draw (7.5,6.5) node (b1) [label=center:$\bullet$] {};
        \draw (7.5,7.5) node (b1) [label=center:$\bullet$] {};
           \draw (6.5,.5) node (b1) [label=center:$\bullet$] {};
         \draw (6.5,1.5) node (b1) [label=center:$\bullet$] {};
         \draw (6.5,2.5) node (b1) [label=center:$\bullet$] {};
         \draw (6.5,3.5) node (b1) [label=center:$\bullet$] {};
         \draw (6.5,4.5) node (b1) [label=center:$\bullet$] {};
         \draw (6.5,5.5) node (b1) [label=center:$\bullet$] {};
         \draw (6.5,6.5) node (b1) [label=center:$\bullet$] {};
        \draw (6.5,7.5) node (b1) [label=center:$\bullet$] {};
           \draw (5.5,.5) node (b1) [label=center:$\bullet$] {};
         \draw (5.5,1.5) node (b1) [label=center:$\bullet$] {};
         \draw (5.5,2.5) node (b1) [label=center:$\bullet$] {};
         \draw (5.5,3.5) node (b1) [label=center:$\bullet$] {};
         \draw (5.5,4.5) node (b1) [label=center:$\bullet$] {};
         \draw (5.5,5.5) node (b1) [label=center:$\bullet$] {};
         \draw (5.5,6.5) node (b1) [label=center:$\bullet$] {};
        \draw (5.5,7.5) node (b1) [label=center:$\bullet$] {};
         \draw (4.5,1.5) node (b1) [label=center:$\bullet$] {};
         \draw (4.5,2.5) node (b1) [label=center:$\bullet$] {};
         \draw (4.5,3.5) node (b1) [label=center:$\bullet$] {};
         \draw (4.5,4.5) node (b1) [label=center:$\bullet$] {};
         \draw (4.5,5.5) node (b1) [label=center:$\bullet$] {};
         \draw (4.5,6.5) node (b1) [label=center:$\bullet$] {};
         \draw (4.5,7.5) node (b1) [label=center:$\bullet$] {};
         \draw (3.5,1.5) node (b1) [label=center:$\bullet$] {};
         \draw (3.5,2.5) node (b1) [label=center:$\bullet$] {};
         \draw (3.5,3.5) node (b1) [label=center:$\bullet$] {};
         \draw (3.5,4.5) node (b1) [label=center:$\bullet$] {};
         \draw (3.5,5.5) node (b1) [label=center:$\bullet$] {};
         \draw (3.5,6.5) node (b1) [label=center:$\bullet$] {};
         \draw (3.5,7.5) node (b1) [label=center:$\bullet$] {};
         \draw (2.5,3.5) node (b1) [label=center:$\bullet$] {};
         \draw (2.5,4.5) node (b1) [label=center:$\bullet$] {};
         \draw (2.5,5.5) node (b1) [label=center:$\bullet$] {};
          \draw (2.5,6.5) node (b1) [label=center:$\bullet$] {};
           \draw (2.5,7.5) node (b1) [label=center:$\bullet$] {};
         \draw (1.5,5.5) node (b1) [label=center:$\bullet$] {};
         \draw (1.5,6.5) node (b1) [label=center:$\bullet$] {};
        \draw (1.5,7.5) node (b1) [label=center:$\bullet$] {};
       \draw (0.5,7.5) node (b1) [label=center:$\bullet$] {};

     \end{tikzpicture}
     }
      \caption{Staircase Condition as in Theorem~\ref{T-Fd2}}
      \label{F-Big}
      \end{figure}
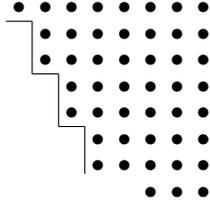
\end{exa}

\begin{rem}\label{R-Staircase}
A particularly nice case of Theorem~\ref{T-Fd2} arises when the last $\delta-2$ columns of~$\cF$ are full (i.e., have~$m$ dots). In this case
$\varepsilon=m-c_{\ell+1}$ and~\eqref{e-ct} reads as $c_t\leq m-(m-c_{\ell+1})(\ell+2-t)$ for $t\in[\ell]$.
\end{rem}

\begin{exa}\label{E-Staircase}
Consider the $6\times6$-Ferrers diagram $\cF=[1,2,4,5,6,6]$, shown in Figure~\ref{F-FStaircase1}, and let $\delta=4$, hence $\ell=n-\delta+1=3$.
     Then $\varepsilon=1$ and we are in the situation of Remark~\ref{R-Staircase}.
     The conditions
     $c_t\leq m-(m-c_4)(\ell+2-t)=6-(6-5)(5-t)=1+t$ for $t=1,\ldots,\ell$ are indeed satisfied and thus
      maximal $[\cF;4]_q$-codes exist over every field~$\F_q$.
           An analogous comment as in Example~\ref{E-StaircaseGeneral} applies to the two pending dots in the last row.
     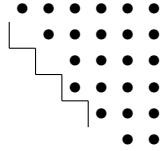
\begin{figure}[ht]
     \centering
     {\small
     \begin{tikzpicture}[scale=0.35]
         \draw (3,1) to (3,2);
         \draw (3,2) to (2,2);
         \draw (2,2) to (2,3);
          \draw (2,3) to (1,3);
           \draw (1,3) to (1,4);
            \draw (1,4) to (0,4);
            \draw (0,4) to (0,5);
              \draw (5.5,.5) node (b1) [label=center:$\bullet$] {};
         \draw (5.5,1.5) node (b1) [label=center:$\bullet$] {};
         \draw (5.5,2.5) node (b1) [label=center:$\bullet$] {};
         \draw (5.5,3.5) node (b1) [label=center:$\bullet$] {};
         \draw (5.5,4.5) node (b1) [label=center:$\bullet$] {};
         \draw (5.5,5.5) node (b1) [label=center:$\bullet$] {};
         \draw (4.5,.5) node (b1) [label=center:$\bullet$] {};
         \draw (4.5,1.5) node (b1) [label=center:$\bullet$] {};
         \draw (4.5,2.5) node (b1) [label=center:$\bullet$] {};
         \draw (4.5,3.5) node (b1) [label=center:$\bullet$] {};
         \draw (4.5,4.5) node (b1) [label=center:$\bullet$] {};
         \draw (4.5,5.5) node (b1) [label=center:$\bullet$] {};
         \draw (3.5,1.5) node (b1) [label=center:$\bullet$] {};
         \draw (3.5,2.5) node (b1) [label=center:$\bullet$] {};
         \draw (3.5,3.5) node (b1) [label=center:$\bullet$] {};
         \draw (3.5,4.5) node (b1) [label=center:$\bullet$] {};
         \draw (3.5,5.5) node (b1) [label=center:$\bullet$] {};
         \draw (2.5,2.5) node (b1) [label=center:$\bullet$] {};
         \draw (2.5,3.5) node (b1) [label=center:$\bullet$] {};
         \draw (2.5,4.5) node (b1) [label=center:$\bullet$] {};
         \draw (2.5,5.5) node (b1) [label=center:$\bullet$] {};
         \draw (1.5,4.5) node (b1) [label=center:$\bullet$] {};
        \draw (1.5,5.5) node (b1) [label=center:$\bullet$] {};
       \draw (0.5,5.5) node (b1) [label=center:$\bullet$] {};
     \end{tikzpicture}
     }
     \caption{Staircase Condition as in Remark~\ref{R-Staircase}}
     \label{F-FStaircase1}
     \end{figure}
\end{exa}

The following is immediate with Remark~\ref{R-Staircase}.

\begin{cor}\label{C-Upper}
Conjecture~\ref{C-FConj} holds true for $n\times n$-upper triangular matrices with $\delta=3$.
\end{cor}

We also obtain an analogue to Corollary~\ref{C-cn}.
It arises as a generalization of the situation discussed in Remark~\ref{R-Staircase}.

\begin{cor}\label{C-ctn}
Let $\ell=n-\delta+1$ and $\cF=[c_1,\ldots,c_n]$ be an $m\times n$-Ferrers diagram such that $c_t\geq n$ for all $t=\ell+2,\ldots,n$
(that is, the last $\delta-2$ columns of~$\cF$ have at least~$n$ dots) and such that
\[
  c_t\leq n-(n-c_{\ell+1})(\ell+2-t) \text{ for } t\in[\ell].
\]
Then all dots at positions $(i,j)$ where $i>\max\{c_{\ell+1},n\}$ are pending and there exists a maximal $[\cF;\delta]_q$-code for any field~$\F_q$.
\end{cor}

\begin{proof}
If $c_{\ell+1}\geq n$, the result is in Corollary~\ref{C-cn}.
Thus let us assume $c_{\ell+1}<n$.
Set
\[
   \hat{c}_t=\min\{c_t,n\}=\left\{\begin{array}{ll} c_t,&\text{if }t\leq \ell+1,\\ n,&\text{if }t\geq \ell+2, \end{array}\right.
\]
and let $\hat{\cF}=[\hat{c}_1,\ldots,\hat{c}_n]$.
Then~$\hat{\cF}$ is an $n\times n$-Ferrers diagram satisfying the staircase condition of Remark~\ref{R-Staircase}.
In particular, $\numin(\hat{\cF};\delta)=\nu_0(\hat{\cF};\delta)$.
Moreover, $\hat{\cF}\subseteq\cF$ and $\nu_0(\hat{\cF};\delta)=\sum_{t=1}^{\ell}\hat{c}_t=\nu_0(\cF;\delta)$.
Thus $\nu_j(\cF;\delta)\geq\nu_j(\hat{\cF};\delta)\geq\numin(\hat{\cF};\delta)=\nu_0(\hat{\cF};\delta)=\nu_0(\cF;\delta)$
for all $j\in\{1,\ldots,\delta-1\}$.
This tells us that a maximal $[\hat{\cF};\delta]_q$-code is also a maximal $[\cF;\delta]_q$-code and the existence of the
former has been established in Theorem~\ref{T-Fd2} and Remark~\ref{R-Staircase}.
\end{proof}

\begin{exa}\label{E-dotsbelown}
Let $\delta=5$ and $\cF=[3,4,5,6,6,7]$.
Then the last~$\delta-2=3$ columns have at least $n=6$ dots and the staircase condition
from Corollary~\ref{C-ctn} is satisfied.
Hence there exists a maximal $[\cF;5]_q$-code over any field~$\F_q$, and the bottom dot is pending.
     \begin{figure}[ht]
     \centering
     {\small
     \begin{tikzpicture}[scale=0.35]
         \draw (2,1) to (2,2);
          \draw (2,2) to (1,2);
           \draw (1,2) to (1,3);
            \draw (1,3) to (0,3);

         \draw (5.5,-.5) node (b1) [label=center:$\bullet$] {};
         \draw (5.5,.5) node (b1) [label=center:$\bullet$] {};
         \draw (5.5,1.5) node (b1) [label=center:$\bullet$] {};
         \draw (5.5,2.5) node (b1) [label=center:$\bullet$] {};
         \draw (5.5,3.5) node (b1) [label=center:$\bullet$] {};
         \draw (5.5,4.5) node (b1) [label=center:$\bullet$] {};
         \draw (5.5,5.5) node (b1) [label=center:$\bullet$] {};

         \draw (4.5,.5) node (b1) [label=center:$\bullet$] {};
         \draw (4.5,1.5) node (b1) [label=center:$\bullet$] {};
         \draw (4.5,2.5) node (b1) [label=center:$\bullet$] {};
         \draw (4.5,3.5) node (b1) [label=center:$\bullet$] {};
         \draw (4.5,4.5) node (b1) [label=center:$\bullet$] {};
         \draw (4.5,5.5) node (b1) [label=center:$\bullet$] {};

        \draw (3.5,0.5) node (b1) [label=center:$\bullet$] {};
         \draw (3.5,1.5) node (b1) [label=center:$\bullet$] {};
         \draw (3.5,2.5) node (b1) [label=center:$\bullet$] {};
         \draw (3.5,3.5) node (b1) [label=center:$\bullet$] {};
         \draw (3.5,4.5) node (b1) [label=center:$\bullet$] {};
         \draw (3.5,5.5) node (b1) [label=center:$\bullet$] {};

        \draw (2.5,1.5) node (b1) [label=center:$\bullet$] {};
         \draw (2.5,2.5) node (b1) [label=center:$\bullet$] {};
         \draw (2.5,3.5) node (b1) [label=center:$\bullet$] {};
         \draw (2.5,4.5) node (b1) [label=center:$\bullet$] {};
         \draw (2.5,5.5) node (b1) [label=center:$\bullet$] {};

         \draw (1.5,2.5) node (b1) [label=center:$\bullet$] {};
       \draw (1.5,3.5) node (b1) [label=center:$\bullet$] {};
         \draw (1.5,4.5) node (b1) [label=center:$\bullet$] {};
        \draw (1.5,5.5) node (b1) [label=center:$\bullet$] {};

       \draw (0.5,3.5) node (b1) [label=center:$\bullet$] {};
        \draw (0.5,4.5) node (b1) [label=center:$\bullet$] {};
       \draw (0.5,5.5) node (b1) [label=center:$\bullet$] {};
     \end{tikzpicture}
     }
     \caption{Staircase Condition as in Corollary~\ref{C-ctn}}
     \end{figure}
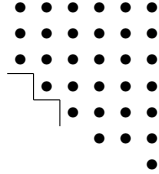
\end{exa}

We close this section with a few instances where a maximum $[\cF;\delta]_q$-code can be realized as an $\F_q$-subspace of an
$\F_{q^m}$-linear $[m\times n;\delta]$-MRD code even though none of the staircase conditions are satisfied.
We need the following lemma.

\begin{lemma}\label{L-MRDSystem}
Given $m\geq n\geq\delta$. Set $\ell=n-\delta+1$.
Furthermore, let $a_1,\ldots,a_\ell\in\F_{q^m}$ be such that $\rk(1,a_1,\ldots,a_\ell)=\ell+1$.
Then there exists a matrix $A\in\F_{q^m}^{\ell\times(n-\ell)}$ such that its first column is given by
$(a_1,\ldots,a_\ell)\T$ and $\cC=\rs (I_\ell\mid A)$ is an $\F_{q^m}$-linear $[m\times n;\delta]$-MRD code.
\end{lemma}

\begin{proof}
Let $G'=(I\mid B)\in\F_{q^m}^{\ell\times n}$ generate an MRD code, and denote the first column of~$B$ by $(b_1,\ldots,b_\ell)\T$.
Then $\rk(1,b_1,\ldots,b_\ell)=\ell+1$ thanks to Lemma~\ref{L-MRDLinInd}.
Hence there exists an $\F_q$-isomorphism $\phi:\F_{q^m}\longrightarrow\F_{q^m}$ such that $\phi(b_i)=a_i$ for $i=1,\ldots,\ell$ and
$\phi(1)=1$.
Set $G=\phi(G')$, where we apply~$\phi$ entrywise to the matrix.
Then $G$ is of the form $G=(I\mid A)$, where the first column of~$A$ is as desired.
Furthermore,~$G$ generates an MRD code.
This follows from the $\F_q$-linearity of~$\phi$ along with the MRD criterion given in \cite[Thm.~2]{Gab85}, which says that a matrix $G\in\F_{q^m}^{\ell\times n}$ generates an MRD code iff
for every $U\in\GL_n(\F_q)$ each maximal minor of $GU$ is nonzero.
\end{proof}

\begin{exa}\label{E-F2244}
Let $\cF=[2,2,4,4]$ and let $\delta=4$. Thus $\numin=2$. Choose an $\F_{q^4}$-linear $[4\times4;4]$-MRD code generated by $G=(1,\beta,\beta',\beta'')\in\F_{q^4}^{1\times 4}$.
Suppose $B=(x_1,x_2,x_3,x_4)$ is a basis such that $\{\phi_B(uG)\mid u\in\subspace{x_1,x_2}\}$ has shape~$\cF$.
The shape implies $\subspace{x_1,x_2}\beta\subseteq\subspace{x_1,x_2}$.
From this one easily derives $\subspace{1,\beta}=\subspace{1,x_1^{-1}x_2}$ as well as $\beta x_1^{-1}x_2\in\subspace{1,x_1^{-1}x_2}$.
In other words, $\beta^2\in\subspace{1,\beta}$.
Such an element clearly exists and any basis of the form $B=(1,\beta,x_3,x_4)$ leads to the desired Ferrers diagram code.
All of this shows that the MRD code generated by $(1,\beta,\beta',\beta'')$ admits a maximal $[\cF;4]$-code iff~$\beta$ has degree~$2$.
We conclude that some, but not every, $\F_{q^4}$-linear $[4\times4;4]$-MRD code contains, for a suitable basis~$B$, a maximal $[\cF;4]_q$-code.
\end{exa}

The following result provides us with maximal Ferrers diagram codes for certain diagrams with at most 3 distinct column indices.
The construction bears some resemblance to~\cite[Thm.~3.2]{ZhGe19}.
However, while the latter requires pending dots for many Ferrers diagrams this is not the case for our construction.
Such an example, not covered by any of the constructions in~\cite{ZhGe19} and not having any pending dots, will be
presented below in Example~\ref{E-F224466}.

\begin{prop}\label{P-invariance}
Let $3\leq\delta\leq n\leq m$ and put $\ell=n-\delta+1$.
Let~$b\in\N$ be a common divisor of~$m$ and $\ell+1$.
Then there exists a maximal $[\cF;\delta]_q$-code, where
\[
    \cF=[\underbrace{b,\ldots,b}_{\ell-b+1},\underbrace{\ell+1,\ldots,\ell+1}_{b},\underbrace{m,\ldots,m}_{\delta-2}].
\]
\end{prop}

\begin{proof}
Let $\alpha\in\F_{q^m}$ be a primitive element, and put $\beta=\alpha^{(q^m-1)/(q^b-1)}$.
Define the $\F_q$-subspace
\[
    V=\Big\langle \alpha^i\beta^j\,\Big|\,0\leq i<\frac{\ell+1}{b},\, 0\leq j<b\Big\rangle_{\F_q}\subset\F_{q^m}.
\]
Using that $b\leq m/2$, one easily verifies that $s:=(\ell+1)/b-1<(q^m-1)/(q^b-1)$, and therefore the described generators form a basis of~$V$.
Thus $\dim_{\F_q}(V)=\ell+1$.
Let~$B$ be a basis of~$\F_{q^m}$, whose first $\ell+1$ elements are the given basis of~$V$ in the order
\begin{equation}\label{e-firstbasisvec}
   1,\beta,\ldots,\beta^{b-1}\,|\, \alpha,\alpha\beta,\ldots,\alpha\beta^{b-1}\,|\,\alpha^2,\alpha^2\beta,\ldots,\alpha^2\beta^{b-1}\,|\ldots\,|\,
   \alpha^s,\alpha^s\beta,\ldots,\alpha^s\beta^{b-1}.
\end{equation}
Note that $\F_q[\beta]=\F_{q^b}$, and thus $\beta^b$ is an $\F_q$-linear combination of $1,\ldots,\beta^{b-1}$.
This in turn implies that~$V$ is $\beta$-invariant.
By Lemma \ref{L-MRDSystem} there exists a matrix $G=(I_{\ell}\mid A)\in\F_{q^m}^{\ell\times n}$ generating an
$\F_{q^m}$-linear MRD code, and where the first column of~$A$ is given by the transpose of
\begin{equation}\label{e-FirstColA}
   (\alpha,\alpha\beta,\ldots,\alpha\beta^{b-1}\,|\,\alpha^2,\alpha^2\beta,\ldots,\alpha^2\beta^{b-1}\,|\ldots\,|\,
   \alpha^s,\alpha^s\beta,\ldots,\alpha^s\beta^{b-1}\,|\,\beta,\ldots,\beta^{b-1}).
\end{equation}
Put
\[
     \cC=\phi_B\big\{(u_1,\ldots,u_{\ell})G\,\big|\,u_1,\ldots,u_{\ell-b+1}\in\big\langle1,\beta,\ldots,\beta^{b-1}\big\rangle
     \text{ and }u_{\ell-b+2},\ldots,u_{\ell}\in V\big\}.
\]
Then $\dim(\cC)=b(\ell-b+1)+(\ell+1)(b-1)=\nu_0(\cF;\delta)$ and~$\cC$ has rank distance $\delta'\geq\delta$.
It remains to see that~$\cC$ is supported on~$\cF$.
This is clearly the case for the first $\ell$ coordinates of any codeword $(u_1,\ldots,u_{\ell})G$ thanks to the choice of~$B$ and~\eqref{e-firstbasisvec}, and it is trivially true for the last $\delta-2$ coordinates.
The $(\ell+1)$-st coordinate is the scalar product of $(u_1,\ldots,u_{\ell})$ and the vector in~\eqref{e-FirstColA}.
By the $\beta$-invariance of~$V$ this product is in~$V$, and thus its coordinate vector has zero entries in the last $m-\ell-1$ positions due to the choice of the basis~$B$.
Thus $\cC$ is an $[\cF;\delta']_q$-code of dimension $\nu_0(\cF;\delta)$.
As in the proof of Theorem~\ref{T-Fn1} this yields the desired result.
\end{proof}

Let us briefly revisit Example~\ref{E-F2244}.
Then~$\cF$ is as in the last proposition ($\ell=1,\,b=2$, and $s=0$), and the case where the entry~$\beta$ of~$G$ has degree~$2$ is the situation from the previous proof.

We conclude this section with an example, which has been mentioned explicitly in \cite[Sec.~VIII]{EGRW16} as an open case, and can now be settled thanks to Proposition~\ref{P-invariance}.

\begin{exa}\label{E-F224466}
Let $m=n=6$ and $\delta=4$. Hence $\ell=n-\delta+1=3$.
Choosing $b=2$ leads to the Ferrers diagram $\cF=[2,2,4,4,6,6]$.
In this case $\numin(\cF;\delta)=8=\nu_j$ for all $j=0,\ldots,3$.
Thus~$\cF$ has no pending dots w.r.t.~$\delta$.
The matrix~$G$ of the previous proof takes the form
\[
   G=\begin{pmatrix}1&0&0&\alpha&b_1&c_1\\ 0&1&0&\alpha\beta&b_2&c_2\\ 0&0&1&\beta&b_3&c_3\end{pmatrix}\in\F_{q^6}^{3\times 6},
\]
where~$\alpha$ is a primitive element of~$\F_{q^6}$ and $\beta:=\alpha^{(q^6-1)/(q^2-1)}$.
The desired maximal $[\cF;4]_q$-code is given by
\begin{equation}\label{e-CC224466}
  \cC:=\big\{\phi_B\big((u_1,u_2,u_3)G\big)\,\big|\, u_1,u_2\in\subspace{1,\beta},\,u_3\in\subspace{1,\beta,\alpha,\alpha\beta} \big\},
\end{equation}
which is indeed $8$-dimensional.
It is worth mentioning that maximal $[\cF;4]_q$-codes are extremely scarce.
Indeed, using SageMath and testing 100,000,000 tuples of~$8$ random matrices of shape~$\cF$ over~$\F_2$ did not lead to a single maximal $[\cF;4]_2$-code.
In Section~\ref{S-Prob} we will discuss more generally the probability that a random selection of $\numin(\cF;\delta)$ matrices
in~$\F_q[\cF]$
generates a maximal $[\cF;\delta]_q$-code.
\end{exa}

\section{Ferrers Diagram Codes not Obtainable from MRD Codes}\label{S-NotFqm}
In Example~\ref{E-F2244} we illustrated that for certain pairs $(\cF;\delta)$ a maximal  $[\cF;\delta]_q$-code can be realized as an $\F_q$-linear
subspace of a suitably chosen $\F_{q^m}$-linear MRD code.
We now present pairs $(\cF;\delta)$ that do not allow the realization of a maximal $[\cF;\delta]_q$-code as a subfield subcode of \emph{any} $\F_{q^m}$-linear MRD code.
In order to do so we need the following simple lemma.

\begin{lemma}\label{L-InvSub}
Let $a\in\F_{q^m}\setminus\F_q$ and suppose there is an $\F_q$-subspace~$V$ of~$\F_{q^m}$ that is invariant under multiplication by~$a$.
Then $\gcd(\dim_{\F_q}\!V,m)>1$.
\end{lemma}

\begin{proof}
Let the subfield $\F_q[a]$ have order $q^r$. Then $r>1$ and $r\mid m$.
By assumption~$V$ is an $\F_q[a]$-subspace of $\F_{q^m}$.
Hence $\dim_{\F_q}V=tr$, where $t:=\dim_{\F_q[a]}V$.
This proves the statement.
\end{proof}

\begin{cor}\label{C-NotFqmSubspace}
Let $\cF=[c_1,\ldots,c_n]$ be an $m\times n$-Ferrers diagram and $2\leq\delta\leq n$. Set $\ell=n-\delta+1$.
Suppose
\[
   c_\ell=c_{\ell+1}<m \ \text{ and }\ \gcd(c_\ell,m)=1.
\]
If $\numin(\cF;\delta)=\nu_0(\cF;\delta)=\sum_{t=1}^\ell c_t$, then a maximal $[\cF;\delta]_q$-code does not exist as an $\F_q$-subspace of an $\F_{q^m}$-linear $[m\times n;\delta]$-MRD code.
\end{cor}

Note that in the situation of this corollary, the step size $\varepsilon=\sum_{t=\ell+1}^n(m-c_t)$ from Theorem~\ref{T-Fd2} is positive and the staircase condition~\eqref{e-ct} is not satisfied.

\begin{proof}
Suppose by contradiction that $G=(I_\ell\mid A)\in\F_{q^m}^{\ell\times n}$ generates an MRD code that contains a maximal $[\cF;\delta]$-code.
This means, there exists a basis $B=(x_1,\ldots,x_m)$ of~$\F_{q^m}$ such that
\[
  \phi_B\big((u_1,\ldots,u_\ell)G\big)\text{ has shape $\cF$ for all $u_t\in\subspace{x_1,\ldots,x_{c_t}},\,t\in[\ell]$}.
\]
Set $V:=\subspace{x_1,\ldots,x_{c_\ell}}$. Then $u_t\in V$ for all $t\in[\ell]$.
Let $uA=(v_1,\ldots,v_{n-\ell})$.
Then $v_1=\sum_{t=1}^\ell u_t a_t$, where $(a_1,\ldots,a_\ell)\T$ is the first column of~$A$, and
$c_\ell=c_{\ell+1}$ implies $v_1\in V$.
Since this has to be true for all choices of~$u_1,\ldots,u_\ell$, we obtain in particular that $u_\ell a_\ell\in V$ for all $u_\ell\in V$ and conclude that $V$ is $a_\ell$-invariant.
By Lemma~\ref{L-MRDLinInd} the element $a_\ell$ is not in~$\F_q$, and thus
Lemma~\ref{L-InvSub} leads to a contradiction to the given coprimeness of~$c_\ell$ and~$m$.
\end{proof}

Now we are ready to present some examples.
\begin{exa}\label{E-F1334a}
For~$\cF=[1,3,3,4]$ and $\delta=3$ we have $\ell=2$ and $c_2=c_3=3$.
Thus, by Corollary~\ref{C-NotFqmSubspace} a maximal $[\cF;3]_q$-code is not realizable as an
$\F_q$-subspace of an $\F_{q^4}$-linear $[4\times 4;3]$-MRD code.
As we saw in Example~\ref{E-F1334}, such codes can nevertheless easily be constructed in an ad-hoc manner.
In Example~\ref{E-F1334b} we will return to this Ferrers diagram and discuss the probability that $4$ randomly chosen matrices in $\F_q[\cF]$
generate  a maximal $[\cF;\delta]_q$-code.
\end{exa}

\begin{exa}\label{E-Fn1again}
Let~$\cF$ be the $5\times5$-Ferrers diagram $\cF=[1,3,4,4,5]$ and $\delta=3$.
Then $\ell=3$ and $\numin(\cF;\delta)=c_1+c_2+c_3=8,\,c_3=c_4=4,\,\gcd(c_3,m)=1$ (and~$\cF$ has no pending dots w.r.t.~$\delta=3$).
Again, Corollary~\ref{C-NotFqmSubspace} implies that a maximal $[\cF;3]$-code cannot be obtained as an $\F_q$-subspace of an
$\F_{q^5}$-linear$[5\times 5;3]$-MRD code.
In this case a maximal $[\cF;\delta]_q$-code can be obtained by~\cite[Construction~2, Thm.~8]{EGRW16}.
The assumptions of~\cite[Thm.~8]{EGRW16} are indeed met:
 (1) the last $\delta-1$ columns have at least $n-1$ dots, (2) the first~$n-\delta+1$ columns have at most~$n-1$ dots\footnote{This assumption is not explicitly mentioned
 in \cite[Construction~2, Thm.~8]{EGRW16} but is in fact necessary; see also the paragraph after the proof of Thm.~8 in~\cite{EGRW16}.}, (3)  $m\geq n-1+c_1$.
\end{exa}

\begin{exa}\label{E-F22555}
Consider the $5\times5$-Ferrers diagram $\cF=[2,2,5,5,5]$ with $\delta=5$ and $\ell=n-\delta+1=1$.
Hence $c_\ell=c_{\ell+1}=2$ and $\numin(\cF;5)=c_1=2$.
Thus, as above, a maximal $[\cF;5]$-code cannot be realized as an $\F_q$-subspace of an
$\F_{q^5}$-linear $[5\times 5;5]$-MRD code.
However, such a code can easily be obtained as follows.
First of all,~$\cF$ has a pending dot at $(5,3)$.
Removing that dot leads to a Ferrers diagram covered by~\cite[Thm.~9]{EGRW16}.
The simple proof shows how to construct the desired maximal $[\cF;5]_q$-code over any field~$\F_q$.
\end{exa}

\section{Upper Triangular Shape and Rank~$n-1$}\label{S-Recurs}

In this short section we establish the existence of maximal $n\times n$-Ferrers diagram codes
of upper triangular shape with rank distance~$n-1$ in two different ways.
The first one is by induction on~$n$ and a pure existence result.
The second one is an explicit construction based on an irreducible polynomial.
We leave it as an open problem whether either construction can be generalized to upper triangular matrices with rank distance $\delta<n-1$.

We start with the recursive construction for which the following lemma is crucial.
We denote the column space of a matrix~$M$ by $\cs(M)$.

\begin{lemma}\label{L-RowSpLem}
Let $\F=\F_q$ and $A,B\in\F^{n\times n}$ be such that $\cs(B)\nsubseteq\cs(A)$.
Then there exist vectors $v,w\in\F^n$ such that for all $(\lambda,\mu)\in\F^2\setminus\{(0,0)\}$
\[
   \rk(\lambda A+\mu B)\leq n-1\Longrightarrow \lambda v+\mu w\not\in\cs(\lambda A+\mu B).
\]
\end{lemma}

\begin{proof}
Choose $v\in\cs(B)\setminus\cs(A)$.
It suffices to show the existence of a vector $w\in\F^n$ such that $\lambda v+w\not\in\cs(\lambda A+ B)$
whenever $\rk(\lambda A+B)\leq n-1$.

To this end, set $M_{\lambda}:=\lambda A+B$ and define $\cM=\{\lambda\in\F\mid \rk(M_{\lambda})\leq n-1\}$.
Moreover, for each $\lambda\in\cM$ define the affine map
\[
   f_{\lambda}:\F^n\longrightarrow\F^n,\quad x\longmapsto M_{\lambda}x-\lambda v.
\]
Then for any $z\in\F^n$ we have $z\in\im(f_{\lambda})\Longleftrightarrow \lambda v+z\in\cs(M_{\lambda})$.
Hence we need to show the existence of a  vector $w\in\F^n\setminus\cJ$, where $\cJ=\bigcup_{\lambda\in\cM}\im(f_{\lambda})$.
Note that $|\cM|\leq q$ and $|\im(f_{\lambda})|\leq q^{n-1}$ for all $\lambda\in\cM$.
Thus $|\cJ|\leq q^n$.
Clearly, if $|\cM|<q$ we have $|\cJ|<q^n$, as desired.
Hence let $\cM=\F_q$.
In this case the union is not disjoint because by choice of~$v$ we have $v=Bx$ for some $x\in\F^n$ and thus
$v=f_0(x)=f_{-1}(0)$.
Thus, again $|\cJ|<q^n$.
\end{proof}

Now we can establish the existence of maximal $[\cF;n-1]_q$-codes for the $n\times n$-upper triangle~$\cF$.

\begin{theo}\label{T-UTn-1}
Let $\cF=[1,2,\ldots,n]$, thus $\F_q[\cF]$ is the space of upper triangular matrices over~$\F_q$.
Let $\delta=n-1$, hence $\numin(\cF;n-1)=3$.
Then for every~$q$  there exists a maximal $[\cF;n-1]_q$-code.
Thus,  Conjecture~\ref{C-FConj} is true for the pair $(\cF;n-1)$.
\end{theo}

\begin{proof}
We induct on~$n$.
For $n=2$ the statement is trivially true since the matrices
\[
A=\begin{pmatrix}1&0\\0&0\end{pmatrix},
B=\begin{pmatrix}0&0\\0&1\end{pmatrix},\text{ and }
C=\begin{pmatrix}0&1\\0&0\end{pmatrix}
\]
generate a 3-dimensional code over any field, and the minimum prescribed distance is only~$1=n-1$.

Suppose now the statement is true for size~$n$ and that $A,B,C$ generate a maximal $[\cF;n-1]_q$-code in $\F^{n\times n}$.
Assume $\cs(B)\nsubseteq\cs(A)$.

By Lemma~\ref{L-RowSpLem} there exist $v,\,w\in\F^n$ such that $\lambda v+\mu w\notin\cs(\lambda A+\mu B)$ whenever $\rk(\lambda A+\mu B)=n-1$.
Define the (upper triangular) matrices
\[
\widehat A=\left(\begin{array}{c|c}A & v\\ \hline 0 & 0\end{array}\right),
\widehat C=\left(\begin{array}{c|c}B & w\\ \hline 0 & 0\end{array}\right),
\widehat B=\left(\begin{array}{c|c}C & 0\\ \hline 0 & 1\end{array}\right)\in\F^{(n+1)\times(n+1)}.
\]
Consider a general linear combination
\[
  \Omega:=\lambda\widehat{A}+\mu\widehat{C}+\nu\widehat{B}
    =\left(\begin{array}{c|c}
        \lambda A+\mu B+\nu C& \lambda v+\mu w\\\hline 0&\nu
   \end{array}\right).
\]
If $\nu\not=0$ then clearly $\rk(\Omega)\geq n$, while for $\nu=0$ the choice of~$v,w$ also guarantees that $\rk(\Omega)=n$.
This shows that $\widehat{A},\widehat{B},\widehat{C}$ generate a maximal $[\widehat{\cF},n]_q$-code in $\F^{(n+1)\times(n+1)}$, where
$\widehat{\cF}=[1,2,\ldots,n+1]$.
Finally note that $\cs(\widehat{B})\nsubseteq\cs(\widehat{A})$, and we may apply the induction step again to this triple of matrices.
\end{proof}

We conclude this section with an explicit construction.
The proof, appearing in~\cite{Ant19}, is straightforward matrix algebra making use of the repeated appearance of the matrix $\Smallfourmat{0}{1}{c}{d}$  in~$A_3$, which has empty spectrum.

\begin{theo}\label{T-UppTriagExpl}
Let $\cF=[1,2,\ldots,n]$ and $x^2-dx-c\in\F_q[x]$ be an irreducible polynomial.
Define the $n\times n$-matrices
\[
  A_1=\begin{pmatrix}0& & & \\ &1& &\\ & &\ddots& \\ & & &1\end{pmatrix},\
  A_2=\begin{pmatrix}0&1& &\\ & & &\ddots& \\ & & & &1\\ & & & &0\end{pmatrix},\
  A_3=\begin{pmatrix}1&d&-1& & & & & &\\ &0&0&1 & & & & & \\ & &c&d&-1& & & &\\ & & &0&0&1 & & & \\ & & & &c&d&-1& & \\  & & & & &\ddots&\ddots&\ddots& \end{pmatrix}.
\]
Then the code $\cC\subseteq\F_q[\cF]$ generated by $A_1,A_2,A_3$ is a maximal $[\cF;n-1]$-code.
\end{theo}

\section{On the Genericity of Maximal Ferrers Diagram Codes}\label{S-Prob}

In this section we study the likelihood that a randomly chosen  Ferrers diagram code of a given dimension has maximum rank.
It will turn out that the answer depends highly on the choice of the Ferrers diagram.
Special attention will be paid to MRD codes.

For MRD codes the question has also been studied in~\cite{NHTRR18} by Neri et al. (focussing on $\F_{q^m}$-linear MRD codes) and
in~\cite{ByRa18} by Byrne/Ravagnani.
We will give more details and compare our results to those as we go along.

As before we assume throughout that $n\leq m$ and $\delta\in[n]$.
We cast the following definition.

\begin{defi}\label{D-Generic}
Let $\cF$ be an $m\times n$-Ferrers diagram and $\delta\in[n]$. Set $N=\numin(\cF;\delta)$.
Then $N\leq |\cF|=\dim\F_q[\cF]$.
Consider the spaces
\[
       T_q=\{\cC\subseteq\F_q[\cF]\mid \dim(\cC)=N\}\ \text{ and }\
       \hat{T}_q=\{\cC\in T_q\mid \dd(\cC)=\delta\},
\]
thus $\hat{T}_q$ is the set of maximal $[\cF;\delta]_q$-codes.
Then the fraction $|\hat{T}_q|/|T_q|$ is called the \emph{proportion of maximal $[\cF;\delta]$-codes} (within the space of all $N$-dimensional subspaces of $\F_q[\cF]$).
We say that maximal $[\cF;\delta]$-codes are \emph{generic} if
\[
            \lim_{q\rightarrow\infty}\frac{|\hat{T}_q|}{|T_q|}=1.
\]
\end{defi}

Of course, investigating genericity does not address the existence of maximal $[\cF;\delta]$-codes over any given finite field.
Note also that maximal $[\cF;1]$-codes are trivially generic.

It will occasionally be useful for us to express genericity in terms of the probability that randomly chosen matrices generate a maximal $[\cF;\delta]$-code.
In order to do so, we need to fix the probability distribution on $\F_q^{m\times n}$ such that
all entries of a matrix $A=(a_{ij})\in\F_q^{m\times n}$  are independent and uniformly distributed.
Thus, for all $(i,j)$ and all $\alpha\in\F_q$:
\[
  \Prob(a_{ij}=\alpha)=q^{-1}.
\]
For a matrix with shape~$\cF$, the above applies to all entries inside~$\cF$ whereas all other entries are zero with probability~$1$.
We say that $A_1,\ldots,A_N\in\F_q[\cF]$ are \emph{randomly chosen matrices} if they are chosen independently and randomly according to the above distribution.
We will frequently, and without specific mention, make use of the well-known identity
\[
   \big|\big\{M\in\F_q^{a\times b}\,\big|\, \rk M=b\big\}\big|=\prod_{i=0}^{b-1}(q^a-q^i).
\]

\begin{prop}\label{P-CodesMatrixLists}
Fix a pair $(\cF;\delta)$ and let $N=\numin(\cF;\delta)$.
Define
\[
      P_q:=\Prob\big(\subspace{A_1,\ldots,A_N} \text{ is an $[\cF,N;\delta]_q$-code}\big)
\]
for randomly chosen matrices $A_1,\ldots,A_N\in\F_q[\cF]$.
Then
\begin{equation}\label{e-Pq}
    \frac{|\hat{T}_q|}{|T_q|}=P_q\!\cdot\!\frac{q^{|\cF|N}}{\prod_{i=0}^{N-1}(q^{|\cF|}-q^i)}.
\end{equation}
As a consequence, $\lim_{q\rightarrow\infty} |\hat{T}_q|/|T_q|= \lim_{q\rightarrow\infty} P_q$
and maximal $[\cF;\delta]$-codes are generic in the sense of Definition~\ref{D-Generic}
iff $\lim_{q\rightarrow\infty}P_q=1$.
\end{prop}

\begin{proof}
In addition to the sets~$T_q$ and~$\hat{T}_q$ from Definition~\ref{D-Generic} define
\begin{equation}\label{e-Wsets}
\left.\begin{array}{rcl}
   W_q&=&\{(A_1,\ldots,A_N)\in\F_q[\cF]^N\mid \dim\subspace{A_1,\ldots,A_N}=N\},\\[.5ex]
   \hat{W}_q&=&\{(A_1,\ldots,A_N)\in W_q\mid \dd\subspace{A_1,\ldots,A_N}=\delta\}.
\end{array}\qquad\right\}
\end{equation}
Due to the uniform probability, the probability $P_q$ is given by
$P_q=|\hat{W}_q|/q^{|\cF|N}$.
Furthermore, each code~$\cC$ in~$T_q$ has $\alpha:=\prod_{i=0}^{N-1}(q^N-q^i)$ ordered bases.
In other words, $|T_q|\alpha=|W_q|$ and $|\hat{T}_q|\alpha=|\hat{W}_q|$ which in turn implies
\begin{equation}\label{e-TW}
   \frac{|\hat{T}_q|}{|T_q|}=\frac{|\hat{W}_q|}{|W_q|}=P_q\frac{q^{|\cF|N}}{|W_q|}.
\end{equation}
Using $|W_q|=\prod_{i=0}^{N-1}(q^{|\cF|}-q^i)$, one arrives at~\eqref{e-Pq}.
The final statements follow from
the fact that the rightmost fraction approaches~$1$ as $q\rightarrow\infty$.
\end{proof}

In the next section we will show that $\F_q$-linear $[m\times n;\delta]$-MRD codes are not generic (unless $n=1$)
and will give an upper bound for the asymptotic probability.
This result in stark contrast to the results in~\cite{NHTRR18} by Neri et al., where $\F_{q^m}$-linear rank-metric
codes in $\F_{q^m}^n$ are considered.
The authors show that $\F_{q^m}$-linear MRD codes are generic within the class of all $\F_{q^m}$-linear rank-metric codes.
Let us illustrate the difference of the two settings for $[m\times n;n]$-MRD codes.
In this case, the $\F_{q^m}$-linear case amounts to the question whether a randomly chosen matrix of the form
\[
    G=(g_1,\ldots,g_{n})\in\F_{q^m}^{1\times n}
\]
generates an MRD code.
This is obviously equivalent to the question whether $g_1,\ldots,g_{n}$ are linearly independent over~$\F_q$.
The probability for this is $(\prod_{i=0}^{n-1}(q^m-q^i))/(q^{mn})$ and tends to~$1$ for~$q\longrightarrow\infty$.
In the matrix version the same reads as follows. Let $C\in\F_q^{m\times m}$ be the companion matrix of a primitive polynomial.
The above asks for the probability that for a randomly chosen matrix~$A\in\F_q^{m\times n}$ the matrices $A,CA,\ldots,C^{m-1}A$ span an $[m\times n;n]$-MRD code.
But the latter is simply equivalent to~$A$ having rank~$n$, which again results in the above given probability.

On the other hand, in the space of all $\F_q$-linear rank-metric codes we have to study the probability that randomly chosen matrices
$A_1,\ldots,A_m\in\F_q^{m\times n}$ generate an $[m\times n;n]$-MRD code, which means that for all
$(\lambda_1,\ldots,\lambda_m)\in\F_q^m\setminus0$ the matrix $\sum_{i=1}^m\lambda_i A_i$ has full rank.
As one may expect, this property is not generic.
We will indeed show this later in Corollary~\ref{C-Fbarn}, and in the next section we will provide upper bounds on the probability.

In~\cite{ByRa18} Byrne/Ravagnani use a combinatorial approach to obtain estimates for the proportion of $\F_q$- and $\F_{q^m}$-linear MRD codes.
In \cite[Cor.~5.5]{ByRa18} they also derive the genericity of  $\F_{q^m}$-linear MRD codes, and in
\cite[Cor.~6.2]{ByRa18} they show that the asymptotic proportion of $\F_q$-linear MRD codes is at most $1/2$.
In Theorem~\ref{T-RandMRD} we will significantly improve upon this upper bound.
It should be mentioned, however, that their approach is far more general and also leads to genericity results of other classes of codes.

We now turn to investigating genericity for general pairs $(\cF;\delta)$.
We show first that genericity is equivalent to the existence of a maximal $[\cF;\delta]$-code over an algebraically closed field.
To do so, we consider the algebraic closure $\Fb$ of~$\F_q$.
Recall that Definition~\ref{D-Shape} -- Theorem~\ref{T-UppB} make sense and are valid for matrices over infinite fields as well.
Similarly, Definition~\ref{D-MaxF} and Remarks~\ref{R-Fsmaller} and~\ref{R-FDotRem} are valid over any field.
We will also need the following result.

\begin{lemma}[\mbox{Schwartz-Zippel Lemma~\cite{Schw79,Zi79}}]\label{L-SchZ}
Let~$\F$ be any field and $f\in\F[x_1,\ldots,x_n]$ be a non-zero polynomial of total degree~$d$.
Let~$S$ be a finite subset of~$\F$ and $s_1,\ldots,s_n$ be independently and uniformly selected from~$S$.
Then
\[
    \Prob\big(f(s_1,\ldots,s_n)=0\big)\leq\frac{d}{|S|}.
\]
\end{lemma}

Now we are ready to state and prove the following.

\begin{theo}\label{T-Fbar}
Fix a prime power~$q$ and let~$\Fb$ be the algebraic closure of $\F:=\F_q$.
Consider an $m\times n$-Ferrers diagram~$\cF$ and some $\delta\in[n]$ such that $\numin(\cF;\delta)>0$.
Let $N\leq\numin(\cF;\delta)$.
The following are equivalent.
\begin{romanlist}
\item There exist $A_1,\ldots,A_N\in\Fb[\cF]$ such that $\subspace{A_1,\ldots,A_N} \text{ is an $[\cF,N;\delta]$-code}$.
\item The set $\left\lbrace (A_1,\ldots,A_N)\in\Fb[\cF]^N\,\big|\, \subspace{A_1,\ldots,A_N} \text{ is an $[\cF,N;\delta]$-code}\right\rbrace$
         is a nonempty Zariski-open set in $\Fb^{Nt}$, where $t=|\cF|$ is the number of dots in~$\cF$.
\item Let $P_{q^r,N}=\Prob\big(\subspace{A_1,\ldots,A_N} \text{ is an $[\cF,N;\delta]_{q^r}$-code}\big)$, where $A_1,\ldots,A_N\in\F_{q^r}[\cF]$ are randomly chosen. Then
        $\lim_{r\rightarrow\infty}P_{q^r,N}=1$.
\end{romanlist}
As a consequence, maximal $[\cF;\delta]$-codes are generic iff there exists a maximal $[\cF;\delta]$-code over any algebraically closed field of positive characteristic.
\end{theo}

One should note that for the equivalence a fixed `base field' $\F_q$ is considered, along with its field extensions and algebraic closure.
Only for the consequence, we need to consider all finite fields due to the definition of genericity.

We believe that the existence of maximal $[\cF;\delta]$-codes over an algebraically closed field does not depend on its characteristic but are not
able to provide a proof  at this point.
Later in Theorem~\ref{T-Fdsolid} we will encounter an instance where the existence only depends on the combinatorics of $(\cF;\delta)$, and not on
the choice of algebraically closed field.

\begin{proof}
All three statements imply that the matrices $A_1,\ldots,A_N$ are linearly independent, thus we have to focus on the rank of their nontrivial linear combinations.
(ii)~$\Rightarrow$~(i) is clear and so is (iii)~$\Rightarrow$(i) because $\F_{q^r}[\cF]\subseteq\Fb[\cF]$ for all $r\in\N$.

For (i)~$\Rightarrow$~(ii) we introduce indeterminates $x_{1,1},\ldots,x_{1,t},\ldots,x_{N,1}$, $\ldots,x_{N,t}$ over~$\Fb$ (hence they are also indeterminates over every subfield $\F_{q^r}$ of~$\Fb$).
Define $A_i\in\F[x_{i,1},\ldots,x_{i,t}]^{m\times n}$ as the matrix with shape~$\cF$ so that the indeterminates
are the entries of~$A_i$ at the positions in~$\cF$ (in some order).
For $\ell=1,\ldots,N$ and further indeterminates $y_1,\ldots,y_N$ set
\[
   A^{(\ell)}(y)=\sum_{\substack{i=1\\ i\neq\ell}}^N y_iA_i+A_{\ell}.
\]
In the polynomial ring $R=\F[y_1,\ldots,y_N,x_{1,1},\ldots,x_{N,t}]$
consider the ideal~$I^{(\ell)}$ generated by the $\delta\times\delta$-minors of~$A^{(\ell)}(y)$.
Define the elimination ideals $I^{(\ell)}_0=I^{(\ell)}\cap\F[x_{1,1},\ldots,x_{N,t}]$ and let
$I_0=I^{(1)}_0\cdot\ldots\cdot I^{(N)}_0$.
Then
\[
   \cV(I_0):=\Big\{ a=(a_{1,1},\ldots,a_{N,t})\in\Fb^{Nt}\,\Big|\,f(a)=0\text{ for all }f\in I_0\Big\}\subseteq\Fb^{Nt}
\]
is the variety of~$I_0$ over $\Fb$.
Thus $\cV(I_0)=\bigcup_{\ell=1}^N\cV\big(I_0^{(\ell)}\big)$ and
\[
   I_0\neq\{0\}\Longleftrightarrow \cV(I_0)\subsetneq\Fb^{Nt}
         \Longleftrightarrow\text{ there exists }(a_{1,1},\ldots,a_{N,t})\in\Fb^{Nt}\setminus\bigcup_{\ell=1}^N\cV\big(I_0^{(\ell)}\big).
\]
The right hand side implies that for the given tuple $(a_{1,1},\ldots,a_{N,t})$ and for all $\ell$ and all $\lambda_1,\ldots,\lambda_N\in\Fb$ with $\lambda_{\ell}=1$
there exists a polynomial $f\in I^{(\ell)}$ such that $f(\lambda_1,\ldots,\lambda_N,,a_{1,1},\ldots,a_{N,t})\neq0$.
This in turn means that for the according matrices $A_1,\ldots,A_N\in\Fb[\cF]$, every nontrivial linear combination $\sum_{\ell=1}^N\lambda_\ell A_\ell$
has at least one nonzero $\delta\times\delta$-minor.
In other words, $\dd(\subspace{A_1,\ldots,A_N})\geq\delta$.
Even more, every point $(a_{1,1},\ldots,a_{N,t})$ in the Zariski-open set $\cZ:=\Fb^{Nt}\setminus\cV(I_0)$ leads to such a tuple of matrices.
Since~(i) guarantees that the set~$\cZ$ is nonempty, the implication (i)~$\Rightarrow$~(ii) follows.

For (i)~$\Rightarrow$~(iii) we consider again the ideal~$I_0$.
As in the previous part, the assumption implies $I_0\neq\{0\}$.
Fix any nonzero polynomial~$f$ in~$I_0$.
Thus $f$ is in $\F[x_{1,1},\ldots,x_{N,t}]\subseteq\Fb[x_{1,1},\ldots,x_{N,t}]$.
Let $A_1,\ldots,A_N\in\F_{q^r}[\cF]$  be randomly chosen matrices and denote their entries at the positions in~$\cF$ by $a_{1,1},\ldots,a_{N,t}$.
The Schwartz-Zippel Lemma~\ref{L-SchZ} tells us that
\[
   \Prob\big(f(a_{1,1},\ldots,a_{N,t})\neq0\big)\geq 1-\frac{\deg(f)}{q^r}.
\]
Since~$f$ does not depend on~$r$, we obtain $\lim_{r\rightarrow\infty}(1-\deg(f)/q^r)=1$.
Finally, $f(a_{1,1},\ldots,a_{N,t})\neq0$ implies $\dd(\subspace{A_1,\ldots,A_N})\geq\delta$, and hence we arrive at~(iii).

The rest of the theorem is clear from the definition of genericity and the fact that all finite fields with the same characteristic have the same algebraic closure (up to isomorphism).
\end{proof}

The theorem provides us with plenty of pairs $(\cF;\delta)$ for which maximal $[\cF;\delta]$-codes are not generic.
The simplest case is arguably when $\cF=[n,\ldots,n]$ is the full $n\times n$-Ferrers diagram and $\delta=n$.
In this case Theorem~\ref{T-Fbar}(i) is not even satisfied for $N=2$ because for every pair of matrices~$A,\,B$ in $\GL_n(\Fb)$ the polynomial
$\det(A+y B)\in\Fb[y]$ has a root in~$\Fb$.
Thus, the theorem tells us that $[n\times n;n]$-MRD codes are not generic.
In the next section we will present upper bounds on the probability
$\Prob\big(\subspace{A_1,\ldots,A_{\nu_{\text{min}}(\cF;\delta}} \text{ is a maximal $[\cF;\delta]_q$-code}\big)$
for various pairs $(\cF;\delta)$ including MRD codes.

We now continue to identify a class of pairs $(\cF;\delta)$ for which maximal $[\cF;\delta]$-codes are generic.
This class appeared already in~\cite{EGRW16,GoRa17} because it allows the construction of maximal Ferrers diagram codes with the aid of MDS block codes.
We follow the line of reasoning in~\cite[Thm.~32, Cor.~33]{GoRa17}.
In particular we need the notion of diagonals in a Ferrers diagram.

\begin{defi}\label{D-Diag}
Consider the set $[m]\times[n]$. For $r\in[m]$ define the $r$-th diagonal as
\[
     D_r=\{(i,j)\mid j-i=n-r\}=\{(i,i+n-r)\mid i=\max\{1,r+1-n\},\ldots,r\}.
\]
\end{defi}

Thus
\begin{align*}
  &D_1=\{(1,n)\},\ D_2=\{(1,n-1),\,(2,n)\},\ \ldots,\ D_n=\{(1,1),(2,2),\ldots,(n,n)\},\\[.5ex]
  &D_{n+1}=\{(2,1),(3,2),\ldots,(n+1,n)\},\ \ldots,\ D_m=\{(m+1-n,1),\ldots,(m,n)\}.
\end{align*}
and  $|D_r|=\min\{r,n\}$.

Later we will intersect these diagonals with a given Ferrers diagram.
For the $5\times4$-Ferrers diagram~$\cF$ in Figure~\ref{F-FDiag} we have $|D_r\cap\cF|=r$ for $r=1,\ldots,4$ and $|D_5\cap\cF|=2$.

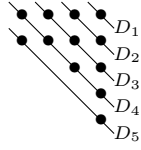
\begin{figure}[ht]
\centering
{\small
\begin{tikzpicture}[scale=0.35]
    \draw (3,6) to (4,5);
    \draw (2,6) to (4,4);
    \draw (1,6) to (4,3);
    \draw (0,6) to (4,2);
    \draw (0,5) to (4,1);
    \draw (4.5,5) node (b1) [label=center:${\scriptstyle D_1}$] {};
    \draw (4.5,4) node (b1) [label=center:${\scriptstyle D_2}$] {};
   \draw (4.5,3) node (b1) [label=center:${\scriptstyle D_3}$] {};
   \draw (4.5,2) node (b1) [label=center:${\scriptstyle D_4}$] {};
   \draw (4.5,1) node (b1) [label=center:${\scriptstyle D_5}$] {};
    \draw (3.5,1.5) node (b1) [label=center:$\bullet$] {};
    \draw (3.5,2.5) node (b1) [label=center:$\bullet$] {};
    \draw (2.5,3.5) node (b1) [label=center:$\bullet$] {};
    \draw (3.5,3.5) node (b1) [label=center:$\bullet$] {};
    \draw (3.5,4.5) node (b1) [label=center:$\bullet$] {};
    \draw (2.5,4.5) node (b1) [label=center:$\bullet$] {};
    \draw (1.5,4.5) node (b1) [label=center:$\bullet$] {};
    \draw (0.5,4.5) node (b1) [label=center:$\bullet$] {};
    \draw (3.5,5.5) node (b1) [label=center:$\bullet$] {};
    \draw (2.5,5.5) node (b1) [label=center:$\bullet$] {};
   \draw (1.5,5.5) node (b1) [label=center:$\bullet$] {};
  \draw (0.5,5.5) node (b1) [label=center:$\bullet$] {};
\end{tikzpicture}
}
\caption{The diagonals of a Ferrers diagram}
\label{F-FDiag}
\end{figure}

In order to cite known results conveniently, we cast the following definition.
The terminology will become clear later.

\begin{defi}\label{D-solid}
Given an $m\times n$-Ferrers diagram~$\cF$ and $\delta\in[n]$, we call $(\cF;\delta)$ \emph{MDS-constructible} if
\[
     \numin(\cF;\delta)=\sum_{i=1}^m\max\{|D_i\cap\cF|-\delta+1,\,0\}.
\]
\end{defi}

Note that only the diagonals of length at least~$\delta$ contribute to the above sum, and therefore
$\sum_{i=1}^m\max\{|D_i\cap\cF|-\delta+1,\,0\}=\sum_{i=\delta}^m\max\{|D_i\cap\cF|-\delta+1,\,0\}$.
We will see in Theorem~\ref{T-Fdsolid} below that this sum is at most $\numin(\cF;\delta)$ for all $(\cF;\delta)$.
The same theorem will show that if $(\cF;\delta)$ is MDS-constructible, then maximal $[\cF;\delta]$-codes are generic.

\begin{exa}\label{E-SF}
\begin{alphalist}
\item Let $a\in\N_0$ and $\cF=[a+1,a+2,\ldots,a+n]$ (hence~$\cF$ is an $n\times n$-upper triangular shape
       with an $a\times n$-rectangle on top). Let $\delta\in[n]$.
       Then for $i=\delta,\ldots,a+n$ we have $|D_i\cap\cF|-\delta+1=\min\{i,n\}-\delta+1$.
       Thus
       \[
           \sum_{i=1}^m\max\{|D_i\cap\cF|-\delta+1,\,0\}=\sum_{i=\delta}^n (i\!-\!\delta\!+\!1)+a(n\!-\!\delta\!+\!1)
           =\frac{(n\!-\!\delta\!+\!1)(n\!-\!\delta\!+\!2)}{2}+a(n\!-\!\delta\!+\!1).
       \]
       On the other hand, it is easy to see that $\numin(\cF;\delta)=\nu_0$ and
       \[
            \nu_0=\sum_{t=a+1}^{n-\delta+1+a}t=\frac{(n-\delta+a+1)(n-\delta+a+2)-a(a+1)}{2},
        \]
        which equals $\sum_{i=1}^m\max\{|D_i\cap\cF|-\delta+1,\,0\}$.
        Thus $(\cF;\delta)$ is MDS-constructible.
\item Let~$\cF$ be the full rectangle, thus $\cF=[m]\times[n]$, and let $\delta\in[n]$.
         Then $\numin=m(n-\delta+1)$ and
         \[
             \sum_{i=\delta}^m\max\{|D_i\cap\cF|\!-\!\delta\!+\!1,\,0\}=\sum_{i=\delta}^n(i\!-\!\delta\!+\!1)\!+\!\sum_{i=n+1}^m (n\!-\!\delta\!+\!1)=
             (n\!-\!\delta\!+\!1)\frac{2m\!-\!n\!-\!\delta\!+\!2}{2}.
         \]
         From this one obtains that $(\cF;\delta)$ is not MDS-constructible for any~$\delta$ unless $n=\delta=1$.
\item Consider $\delta=3$ and $\cF=[1,2,2,4,7]$. Then $\numin=5$ and $(\cF;3)$ is MDS-constructible.
        This is shown in the left diagram of Figure~\ref{F-FerrersFound}.
        We show all diagonals~$D_i$ for which $|D_i\cap\cF|\geq\delta$.
        On the other hand, for $\delta=4$ and $\cF'=[2,2,4,4,6]$ we have $\numin=4$, and $(\cF';4)$ is not MDS-constructible.
        The diagram is shown on the right hand side of Figure~\ref{F-FerrersFound}.

\begin{figure}[ht]
\centering
{\small
\begin{tikzpicture}[scale=0.35]
    \draw (-1,6) to (4,1);
    \draw (0,6) to (4,2);
    \draw (1,6) to (4,3);
 \draw (4.5,3) node (b1) [label=center:${\scriptstyle D_3}$] {};
  \draw (4.5,2) node (b1) [label=center:${\scriptstyle D_4}$] {};
   \draw (4.5,1) node (b1) [label=center:${\scriptstyle D_5}$] {};
   \draw (3.5,-0.5) node (b1) [label=center:$\bullet$] {};
   \draw (3.5,0.5) node (b1) [label=center:$\bullet$] {};
    \draw (3.5,1.5) node (b1) [label=center:$\bullet$] {};
     \draw (2.5,2.5) node (b1) [label=center:$\bullet$] {};
    \draw (3.5,2.5) node (b1) [label=center:$\bullet$] {};
    \draw (2.5,3.5) node (b1) [label=center:$\bullet$] {};
    \draw (3.5,3.5) node (b1) [label=center:$\bullet$] {};
    \draw (3.5,4.5) node (b1) [label=center:$\bullet$] {};
    \draw (2.5,4.5) node (b1) [label=center:$\bullet$] {};
    \draw (1.5,4.5) node (b1) [label=center:$\bullet$] {};
    \draw (0.5,4.5) node (b1) [label=center:$\bullet$] {};
    \draw (3.5,5.5) node (b1) [label=center:$\bullet$] {};
    \draw (2.5,5.5) node (b1) [label=center:$\bullet$] {};
   \draw (1.5,5.5) node (b1) [label=center:$\bullet$] {};
  \draw (0.5,5.5) node (b1) [label=center:$\bullet$] {};
   \draw (-0.5,5.5) node (b1) [label=center:$\bullet$] {};
\end{tikzpicture}
\qquad\qquad
\begin{tikzpicture}[scale=0.35]
       \draw (-1,7) to (4,2);
    \draw (0,7) to (4,3);
  \draw (4.5,3) node (b1) [label=center:${\scriptstyle D_4}$] {};
   \draw (4.5,2) node (b1) [label=center:${\scriptstyle D_5}$] {};

    \draw (3.5,0.5) node(b1) [label=center:\phantom{$bullet$} ]{};

   \draw (3.5,1.5) node (b1) [label=center:$\bullet$] {};
    \draw (3.5,2.5) node (b1) [label=center:$\bullet$] {};

    \draw (1.5,3.5) node (b1) [label=center:$\bullet$] {};
     \draw (2.5,3.5) node (b1) [label=center:$\bullet$] {};
    \draw (3.5,3.5) node (b1) [label=center:$\bullet$] {};

    \draw (1.5,4.5) node (b1) [label=center:$\bullet$] {};
    \draw (2.5,4.5) node (b1) [label=center:$\bullet$] {};
    \draw (3.5,4.5) node (b1) [label=center:$\bullet$] {};

    \draw (3.5,5.5) node (b1) [label=center:$\bullet$] {};
    \draw (2.5,5.5) node (b1) [label=center:$\bullet$] {};
    \draw (1.5,5.5) node (b1) [label=center:$\bullet$] {};
    \draw (0.5,5.5) node (b1) [label=center:$\bullet$] {};
    \draw (-0.5,5.5) node (b1) [label=center:$\bullet$] {};

    \draw (3.5,6.5) node (b1) [label=center:$\bullet$] {};
    \draw (2.5,6.5) node (b1) [label=center:$\bullet$] {};
   \draw (1.5,6.5) node (b1) [label=center:$\bullet$] {};
  \draw (0.5,6.5) node (b1) [label=center:$\bullet$] {};
   \draw (-0.5,6.5) node (b1) [label=center:$\bullet$] {};
\end{tikzpicture}
}
\caption{$(\cF;3)$ is MDS-constructible and $(\cF';4)$ is not MDS-constructible}
\label{F-FerrersFound}
\end{figure}
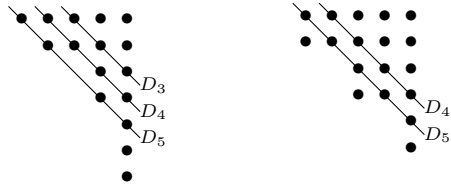
\end{alphalist}
\end{exa}

Now we can formulate a particular construction of maximum $[\cF;\delta]$-codes over sufficiently large finite fields.
It appears in~\cite[Thm.~32]{GoRa17} but actually goes already back to~\cite[p.~329]{Ro91}.
We include the case of algebraically closed fields and present the proof in the appendix.
The construction is based on placing the codewords of suitable MDS-block codes on the diagonals, thus our terminology MDS-constructible.

\begin{theo}\label{T-Fdsolid}
Consider an $m\times n$-Ferrers diagram~$\cF$ and $\delta\in[n]$.
Then one can construct an $[\cF;\delta]$-code of dimension at least $\sum_{i=\delta}^m\max\{|D_i\cap\cF|-\delta+1,\,0\}$
over any field of size at least $\max\{|D_i\cap\cF|-1\mid i=\delta,\ldots,m\}$ (including infinite fields).
Hence
\[
    \numin(\cF;\delta)\geq\sum_{i=\delta}^m\max\{|D_i\cap\cF|-\delta+1,\,0\}.
\]
As a consequence, if $(\cF;\delta)$ is MDS-constructible there exists a maximal $[\cF;\delta]$-code over any algebraically closed field
and thus maximal $[\cF;\delta]$-codes are generic.
\end{theo}

\begin{exa}\label{E-F1334c}
\begin{alphalist}
\item For the pairs $(\cF;\delta)$ discussed in Example~\ref{E-SF}(a) and~(c), $[\cF;\delta]$-codes are generic.
\item Consider $\cF=[1,3,3,4]$ and $\delta=3$; see Example~\ref{E-F1334}. Then $(\cF;3)$ is not MDS-constructible, and
         in Corollary~\ref{C-Prob1334} we will see that maximal $[\cF;3]$-codes are not generic.
\end{alphalist}
\end{exa}

We now turn to the special case where $\delta=n$.
We make use of another result by Gorla/Ravag\-nani~\cite{GoRa17}.

\begin{theo}[\mbox{\cite[Thm.~16]{GoRa17}}]\label{T-GoRaFbar}
Let $\Fb$ be an algebraically closed field and $\cF=[c_1,\ldots,c_n]$. Set
\[
   c:=\min\{c_t-t+1\mid t=1,\ldots,n\}.
\]
Then the maximum possible dimension of an $[\cF;n]$-code over~$\Fb$ is~$\max\{c,0\}$.
Thanks to Theorem~\ref{T-UppB} we thus have $c\leq\numin(\cF;n)$.
\end{theo}

Thus, by Theorem~\ref{T-UppB} maximal $[\cF;n]$-codes over~$\Fb$ exist iff $c=\numin(\cF;n)$.
This occurs only in exceptional cases.
Part~(a) of the next theorem deals with the case that $\numin(\cF;n)$ is attained by $\nu_j(\cF;n)$ for some $j>0$ (and possibly also by $\nu_0(\cF;n)$).
In this case there exists a maximal $[\cF;n]$-code over $\Fb$ exactly in the trivial case where $\numin(\cF;n)=1$.
Part~(b) concerns the case where $\numin(\cF;n)$ is attained  exclusively by $\nu_0(\cF;n)$ and thus equals~$c_1$.
In this case there exists a maximal $[\cF;n]$-code over $\Fb$ iff the Ferrers diagram extends to or below the diagonal that
starts at position $(c_1,1)$, which is~$D_{n+c_1-1}$.

\begin{theo}\label{T-dnumin}
Let $\Fb$ be an algebraically closed field and $\cF=[c_1,\ldots,c_n]$.
Then
\[
      \numin(\cF;n)=0\Longleftrightarrow c_t<t \text{ for some }t\in[n].
\]
Suppose $c_t\geq t$ for all $t\in[n]$.
\begin{alphalist}
\item Suppose $\numin(\cF;n)=\nu_j(\cF;n)$ for some $j>0$.
        Then the following are equivalent.
        \begin{romanlist}
        \item $(\cF;n)$ is MDS-constructible.
        \item There exists a maximal $[\cF;n]$-code over $\Fb$.
        \item $\numin(\cF;n)=1$.
        \item There exists $s\in[n]$ such that $c_s=s$ and $c_t\leq s-1$ for $t=1,\ldots,s-1$ (thus $c_{s-1}=s-1$).
        \end{romanlist}
        Moreover,~(iv) implies $\numin(\cF;n)=\nu_{s-1}(\cF;n)$.
        Finally, if $m>n$ and $c_t\geq m-n+t$ for all~$t$, then $(\cF;n)$ is not MDS-constructible.
\item Suppose $\numin(\cF;n)=\nu_0(\cF;n)<\nu_j(\cF;n)$ for all $j>0$.
        Then the following are equivalent.
        \begin{romanlist}
        \item $(\cF;n)$ is MDS-constructible.
        \item There exists a maximal $[\cF;n]$-code over $\Fb$.
        \item $c_1=\min\{c_t-t+1\mid t=1,\ldots,n\}$ (and thus $c_1\leq m-n+1$).
        \end{romanlist}
\end{alphalist}
\end{theo}

\begin{proof}
The equivalence is immediate with Definition~\ref{D-numin} applied to $\delta=n$.
Let now $c_t\geq t$ for all~$t$. Hence $c>0$ for~$c$ as in Theorem~\ref{T-GoRaFbar}.

(a)
The implication (i)~$\Rightarrow$~(ii) is in Theorem~\ref{T-Fdsolid} and (iii)~$\Rightarrow$~(ii) is trivial.
(iv)~$\Rightarrow$~(iii) follows from
\[
  \nu_{s-1}(\cF;n)=\sum_{t=1}^{s}\max\{c_t-s+1,0\}=\sum_{t=1}^{s-1}\max\{c_t-s+1,0\}+c_s-s+1=1
\]
along with $\numin(\cF;n)\geq1$.
In particular, $\numin(\cF;n)=\nu_{s-1}(\cF;n)$.

In order to show (ii)~$\Rightarrow$~(iv) let~$r$ be such that $c_r-r+1=\min\{c_t-t+1\mid t=1,\ldots,n\}$.
Theorem~\ref{T-GoRaFbar}  tells us that the maximum dimension of an $[\cF;n]$-code over~$\Fb$ is given by $c_r-r+1$.
Note that~(ii) means that $\numin(\cF;n)=c_r-r+1$.
Let~$j>0$ such that $\numin(\cF;n)=\nu_j(\cF;n)$. Then
\[
  c_r-r+1=\nu_j(\cF;n)=\sum_{t=1}^{j}\max\{c_t\!-\!j,0\}+c_{j+1}\!-\!j\geq c_{j+1}\!-\!j=c_{j+1}\!-\!(j\!+\!1)\!+\!1\geq c_r\!-\!r\!+\!1.
\]
Thus we have equality everywhere. In particular, the second inequality yields $c_{j+1}-(j+1)=c_r-r$.
The first inequality implies that $\sum_{t=1}^{j}\max\{c_t\!-\!j,0\}=0$, which in turn means that
$c_t\leq j$ for $t\in[j]$.
Since $j>0$ this is not a vacuous  statement and thus $c_j\leq j$ for some~$j$.
Now the definition of~$r$ yields $c_r-r\leq c_{j}-j\leq 0$.
Thus $c_r=r$ and $c_j=j$ as well as $c_{j+1}=j+1$, and~(iv) follows for $s=j+1$.

It remains to show that~(ii) implies~(i).
Consider the diagonal $D_n=\{(t,t)\mid t=1,\ldots,n\}$.
The assumption $c_t\geq t$ for all~$t$ shows that the dots at $(c_t,t)$ are all on or below this diagonal.
Therefore, $|D_n\cap\cF|=|D_n|=n$.
As a consequence,  $\sum_{i=n}^m\max\{|D_i\cap\cF|-n+1,\,0\}\geq 1=\numin(\cF;n)$.
Hence Theorem~\ref{T-Fdsolid} implies equality, as desired.

Finally, the consequence for $m>n$ follows from the contradiction $1=\numin(\cF;n)=\min\{c_t-t+1\mid t=1,\ldots,n\}\geq m-n+1\geq2$.

(b)
Again,~(i)~$\Rightarrow$~(ii) is in Theorem~\ref{T-Fdsolid}.
For (ii)~$\Rightarrow$~(iii) we note that $\nu_0=c_1$.
Hence there exists a $c_1$-dimensional  $[\cF;n]$-code over~$\Fb$ and
Theorem~\ref{T-GoRaFbar} implies~(iii).
It remains to show (iii)~$\Rightarrow$~(i).
Consider the diagonal $D_{n+c_1-1}=\{(c_1+t-1,t)\mid t=1,\ldots,n\}$. It contains the dot of~$\cF$ at $(c_1,1)$.
Furthermore, since $c_t\geq c_1+t-1$ for all~$t$, the dots of~$\cF$ at positions~$(c_t,t)$ are on or below $D_{n+c_1-1}$ for all~$t$.
Thus $|D_{n+c_1-1}\cap\cF|=n$.
Thanks to the top-alignedness of~$\cF$ we obtain $|D_i\cap\cF|=n$ for all $i=n,\ldots,n+c_1-1$.
This shows that $\sum_{i=n}^m\max\{|D_i\cap\cF|-n+1,\,0\}\geq c_1=\numin(\cF;n)$.
Hence $(\cF;n)$ is MDS-constructible.
\end{proof}

The previous result generalizes the scenario used in the proof of \cite[Prop.~17]{GoRa17}.
Here is a case different from that scenario.

\begin{exa}\label{E-F4466}
Let $\cF=[4,4,6,6]$ and $\delta=4$. Then $\numin(\cF;4)=4=c_1<\nu_j(\cF;4)$ for $j=1,2,3$.
Furthermore, 
$c_1\not\leq m-n+1$,
and therefore there exists no maximal $[\cF;4]$-code over~$\Fb$ by Theorem~\ref{T-dnumin}(b).
As a consequence, maximal $[\cF;4]$-codes are not generic, which means that the probability of the event ``$4$ randomly chosen matrices of shape~$\cF$ generate
a maximal $[\cF;4]_q$-code'' is bounded away from~$1$ (for growing~$q$).
We can be more precise.
If $A_1,\ldots,A_4\in\F_q[\cF]$ generate a  maximal $[\cF;4]_q$-code, then their submatrices consisting of the first~$4$ rows and first~$2$ columns
generate a $[4\times2;2]_q$-MRD code.
In Proposition~\ref{P-SmallMRD} we will see that this happens with a probability less than $0.375$.
Thus the latter is an upper bound for the probability of maximal $[\cF;4]_q$-codes.
\end{exa}

It is worth noting that in both parts of Theorem~\ref{T-dnumin} the existence of a maximal $[\cF;n]$-code over \emph{some} algebraically closed field implies the existence of
a maximal $[\cF;n]$-code over \emph{any} finite field~$\F_q$.
This is obvious in the situation of Theorem~\ref{T-dnumin}(a) because $\numin=1$, and
for the case in~\ref{T-dnumin}(b) the existence of maximal $[\cF;n]$-codes has been established in Example~\ref{E-CompMatMRD}(c).

We summarize the previous results.

\begin{cor}\label{C-Fbarn}
Let $\cF=[c_1,\ldots,c_n]$ be an $m\times n$-Ferrers diagram. Then
\[
  \text{maximal $[\cF;n]$-codes are generic}\Longleftrightarrow \text{$(\cF;n)$ is MDS-constructible}.
\]
In particular, maximal $[m\times n;n]$-MRD codes are not generic whenever $n>1$ (see also Example \ref{E-SF}(b)).
Moreover, if $(\cF;n)$ is MDS-constructible then $\numin(\cF;n)=1$ or $\numin(\cF;n)=\nu_0(\cF;n)=c_1<\nu_j(\cF;n)$ for all $j>0$.
\end{cor}

We strongly believe that the equivalence is true for any rank~$2\leq\delta\leq n$ and note that
``$\Leftarrow$'' has already been established in Theorem~\ref{T-Fdsolid}.
The case of general rank~$\delta$ is much more interesting than $\delta=n$ as it allows for more MDS-constructible pairs $(\cF;\delta)$;
see the left Ferrers diagram in Figure~\ref{F-FerrersFound}.

We conclude this section with the following observation.
In~\cite[Thm.~7]{EGRW16} Etzion et al. provide in essence the same construction of maximal $[\cF;\delta]$-codes over sufficiently large fields as in Theorem~\ref{T-Fdsolid}.
However, their assumption is, on first sight, different from $(\cF;\delta)$ being MDS-constructible.
In the appendix we show that these assumptions are actually equivalent.

\section{Proportions of Nongeneric Ferrers Diagram Codes}\label{S-ProbNonGen}
In this section we focus on the non-generic case and provide some upper bounds on the proportion of maximal$[\cF;\delta]_q$-codes.
In particular, in Theorem~\ref{T-RandMRD} we provide an upper bound for the proportion of MRD codes, which is exact for $[m\times2;2]_q$-MRD codes; see
Corollary~\ref{C-SmallMRD}.
These two results improve on \cite[Cor.~6.2]{ByRa18}, where Byrne/Ravagnani show that the asymptotic proportion is upper bounded by $1/2$.

The main tool in our considerations is the following result about spectrum-free matrices.

\begin{theo}\label{T-SpecFree}
The spectrum of a matrix $A\in\F_q^{n\times n}$ is defined as $\sigma(A)=\{\lambda\in\F_q\mid \lambda\text{ is an eigen-}$ $\text{value of }A\}$.
We call~$A$ spectrum-free if $\sigma(A)=\emptyset$.
Set
\[
  s_n(q)=|\{A\in\F_q^{n\times n}\mid \sigma(A)=\emptyset\}|.
\]
Set $\gamma_n(q)=|\GL_n(\F_q)|=\prod_{j=0}^{n-1}(q^n-q^j)$ and $a_0(q)=1$ and
$a_j(q)=(-1)^j\prod_{\ell=1}^j\frac{1}{q^\ell-1}$ for $j\geq 1$.
Then the generating function of $s_n(q)/\gamma_n(q)$ satisfies
\begin{equation}\label{e-snseries}
     1+\sum_{n=1}^\infty \frac{s_n(q)}{\gamma_n(q)}u^n=\frac{1}{1-u}\prod_{r\geq 1}\left(1-\frac{u}{q^r}\right)^{q-1}
\end{equation}
and
\begin{equation}\label{e-snexplicit}
    s_n(q)=\gamma_n(q)\bigg(\sum_{j=0}^n\sum_{\substack{i_1,\ldots,i_{q-1}\in\N_0:\\ i_1+\ldots+i_{q-1}=j}}a_{i_1}(q)\cdot\ldots\cdot a_{i_{q-1}}(q)\bigg).
\end{equation}
Furthermore,  the proportion of spectrum-free matrices in~$\F_q^{n\times n}$ behaves as follows:
\begin{equation}\label{e-snlimit}
  \lim_{n\rightarrow\infty}\frac{s_n(q)}{q^{n^2}}=\prod_{r\geq1}\left(1-\frac{1}{q^r}\right)^q
  = \lim_{n\rightarrow\infty}\left(\frac{\gamma_n(q)}{q^{n^2}}\right)^q
    \quad \text{ and }\quad \lim_{q\rightarrow\infty}\frac{s_n(q)}{q^{n^2}}=\sum_{j=0}^n\frac{(-1)^j}{j!}.
\end{equation}
\end{theo}

One may note that the expression for $a_j(q)$ can be rewritten as $a_j(q)=1/(q;q)_j$, where $(q;q)_j$ is the $q$-Pochhammer symbol.

\begin{proof}
The main parts of the statements above are in~\cite{Sto88} and~\cite{Mor06}:
Identity~\eqref{e-snseries} is given in \cite[p.~7]{Mor06} and~\eqref{e-snexplicit} appears in~\cite[p.~176]{Sto88}.
As for the limits in~\eqref{e-snlimit}, note that~\eqref{e-snseries} leads to
$\lim_{n\rightarrow\infty}\frac{s_n(q)}{\gamma_n(q)}=\prod_{r\geq1}\Big(1-\frac{1}{q^r}\Big)^{q-1}$  (see also \cite[p.~8]{Mor06}).
On the other hand, clearly $\frac{\gamma_n(q)}{q^{n^2}}=\prod_{r=1}^n\Big(1-\frac{1}{q^{r}}\Big)$.
Taking the limit for $n\rightarrow\infty$  leads to the first parts of~\eqref{e-snlimit}.

It remains to determine  $\lim_{q\rightarrow\infty}\frac{s_n(q)}{q^{n^2}}$.
Note first that
\[
    \lim_{q\to\infty}\frac{s_n(q)}{q^{n^2}}=\lim_{q\to\infty}\frac{s_n(q)}{\gamma_n(q)}=\lim_{q\to\infty}\sum_{j=0}^n b_j(q),\ \text{ where }
   b_j(q)=\sum_{i_1+\ldots+i_{q-1}=j}a_{i_1}(q)\cdots a_{i_{q-1}}(q).
\]
Hence it suffices to consider $\lim_{q\rightarrow\infty}b_j(q)$.
To do so, we need the type of the weak compositions involved in the definition of~$b_j(q)$.
We say that a weak composition $i_1+\ldots+i_{q-1}=j$ is of type $(t_1,\ldots, t_j)$ if $t_k=|\{l\in[q-1]:i_l=k\}|$ for $k\in[j]$.
Then the number of weak compositions of~$j$ of type $(t_1,\ldots,t_j)$ is given by
$\frac{\prod_{i=1}^{t_1+\ldots+t_j}(q-i)}{t_1!\cdot\ldots\cdot t_j!}$, and thus
\[
   b_j(q)=\sum_{t_\ell\in\N_0:\,t_1+2t_2+\ldots+jt_j=j} c(t_1,\ldots,t_j;q),
\]
where
\[
  c(t_1,\ldots,t_j;q):
   =\frac{\displaystyle\prod_{i=1}^{t_1+\ldots+t_j}(q-i)}{t_1!\cdot\ldots\cdot t_j!}\cdot \prod_{k=1}^j(a_k(q))^{t_k}
   =\frac{\displaystyle\prod_{i=1}^{t_1+\ldots+t_j}(q-i)}{t_1!\cdot\ldots\cdot t_j!}\cdot \prod_{k=1}^j\left(\frac{(-1)^k}{\prod_{i=1}^k(q^i-1)}\right)^{t_k}.
\]
As a polynomial in~$q$, the degree of the numerator of $c(t_1,\ldots,t_j;q)$ is $t_1+\ldots+t_j$, and the degree of the denominator is $\sum_{k=1}^j t_k\binom{k+1}{2}$.
Notice that
\[
  \sum_{k=1}^j t_k\binom{k+1}{2} \geq\sum_{k=1}^j t_k\cdot k =j\geq t_1+\ldots +t_j
\]
with equality in both steps  if and only if $t_1=j$ and $t_2=\ldots=t_j=0$.
Therefore
\[
      \lim_{q\to\infty}b_j(q)
      =\lim_{q\to\infty}c(j,0,\ldots,0;q)
      =\lim_{q\to\infty}\frac{\prod_{i=1}^j(q-i)}{j!}\cdot\left(\frac{-1}{q-1}\right)^j
      =\frac{(-1)^j}{j!}.
\]
All of this shows that
$\lim_{q\to\infty}\frac{s_n(q)}{q^{n^2}}
=\lim_{q\to\infty}\sum_{j=0}^n b_j(q)
=\sum_{j=0}^n\frac{(-1)^j}{j!}.
$
\end{proof}

Let us have a closer look at the limit in~\eqref{e-snlimit}.

\begin{rem}\label{R-snlimit}
\begin{alphalist}
\item The infinite product $\pi(q):=\prod_{r\geq1}\Big(1-\frac{1}{q^r}\Big)^q$ takes, for instance, the following approximate values:
        \[
          \begin{array}{c||c|c|c|c|c} q&2&3&5&31&179\\\hline\hline
                \pi(q)&0.0833986&0.175735&0.254108&0.349996& 0.364794\end{array}
        \]
        It is not hard to show that
        \[
            \lim_{q\rightarrow\infty}(1-1/q^r)^q=\left\{\begin{array}{cl}1/e,&\text{if }r=1,\\ 1,&\text{if }r>1,\end{array}\right.
        \]
        and thus $\lim_{q\rightarrow\infty}\pi(q)=1/e\approx0.36788$.
\item By~\eqref{e-snlimit} we may approximate $\frac{s_n(q)}{q^{n^2}}$ by $\big(\frac{\gamma_n(q)}{q^{n^2}}\big)^q$. This is already a
        very good approximation for small values of~$n$ (for instance, $\left|\left(\frac{\gamma_n(q)}{q^{n^2}}\right)^q-\frac{s_n(q)}{q^{n^2}}\right|\leq0.000081$ for $n=7$ and $q=3$).
        Since $\frac{s_n(q)}{q^{n^2}}$ is the fraction of spectrum-free matrices and $\left(\frac{\gamma_n(q)}{q^{n^2}}\right)^q$ the fraction of
        $q$-tuples of invertible matrices within $(\F_q^{n\times n})^q$, the approximation may be interpreted as follows:
        for any randomly chosen matrices $A$ and $A_1,\ldots,A_q\in\F_q^{n\times n}$
        \[
           \Prob(\lambda I-A\text{ is nonsingular for all }\lambda\in\F_q)\approx\Prob(A_1,\ldots,A_q\text{ are nonsingular}).
        \]
        That is, the $q$ dependent matrices $\lambda I-A,\,\lambda\in\F_q,$ behave just like~$q$ independently chosen matrices $A_1,\ldots,A_q$ (with respect to nonsingularity).
        However, computer experiments show that for two randomly chosen matrices $A,B\in\F^{n\times n}$ the probability that all $q^2+q+1$ dependent matrices $\lambda I+\alpha A+\beta B$,
        where the first nonzero coefficient is normalized to~$1$, are nonsingular is much larger than $\left(\frac{\gamma_n(q)}{q^{n^2}}\right)^{q^2+q+1}$.
\end{alphalist}
\end{rem}

We turn now to MRD codes.
We start with the case of $[m\times 2;2]$-MRD codes.
In this case $\numin=\nu_0=m<\nu_1$ and therefore Theorem~\ref{T-dnumin}(b) tells us that there exists no
$[m\times2;2]$-MRD code over an algebraically closed field (which can also be seen from the proof below as there are no spectrum-free matrices over an algebraically closed field).
Thus  $[m\times 2;2]$-MRD codes are not generic.
In order to present an interval for the according probability, we will first consider normalized matrices in the sense described next,
and thereafter relate the result to the proportion of MRD codes in the sense of Definition~\ref{D-Generic}.
Interestingly enough, we will see below that even though there are no MRD codes over the algebraic closure, the probability
does not approach zero for growing field size.

\begin{prop}\label{P-SmallMRD}
Let $\F=\F_q$ and
\[
   A_1=\begin{pmatrix}1&a_{1}^1\\ 0&a_2^1\\ \vdots&\vdots\\0&a_m^1\end{pmatrix},\
   A_2=\begin{pmatrix}0&a_{1}^2\\ 1&a_2^2\\ \vdots&\vdots\\0&a_m^2\end{pmatrix},\  \ldots,
    A_m=\begin{pmatrix}0&a_{1}^m\\ 0&a_2^m\\ \vdots&\vdots\\1&a_m^m\end{pmatrix}\in\F^{m\times2},
\]
where $a_1^1,\ldots,a_m^m$ are randomly chosen field elements.
Set $\cC=\subspace{A_1,\ldots,A_m}$.
Then
\[
  \Prob\big(\cC\text{ is an  $[m\times 2;2]$-MRD code}\big)=\frac{s_m(q)}{q^{m^2}}.
\]
As a consequence, as $q\rightarrow\infty$ the probability approaches $\sum_{j=0}^m\frac{(-1)^j}{j!}$, which is
in the interval $[0.333, 0.375]$ for all $m\geq3$.
\end{prop}

\begin{proof}
Recall that an $[m\times 2;2]$-MRD code has dimension~$m$.
Clearly, the code~$\cC$ given in the proposition has dimension~$m$ and therefore we only have to discuss the rank distance.
A general linear combination of the given matrices has the form
\[
  A(\lambda):=\sum_{\alpha=1}^m\lambda_\alpha A_\alpha=
  \begin{pmatrix}\lambda_1&\sum_{\alpha=1}^ma_1^\alpha\lambda_\alpha\\\vdots&\vdots\\\lambda_m&\sum_{\alpha=1}^ma_m^\alpha\lambda_\alpha\end{pmatrix}.
\]
Thus $A(\lambda)=(\lambda\mid M\lambda)$, where $\lambda=(\lambda_1,\ldots,\lambda_m)\T$ and
\begin{equation}\label{e-MatM}
    M=\begin{pmatrix}a_1^1&\cdots&a_1^m\\ \vdots& &\vdots\\a_m^1&\cdots&a_m^m\end{pmatrix}\in\F^{m\times m}.
\end{equation}
As a consequence,
\[
    \rk(A(\lambda))=2 \text{ for all }\lambda\in\F^m\setminus 0\Longleftrightarrow \sigma(M)=\emptyset.
\]
Now the result follows from the definition of~$s_m(q)$ in Theorem~\ref{T-SpecFree} and from (\ref{e-snlimit}).
\end{proof}

In order to relate the above probability, based on a sample space of normalized matrices, to the proportion of
MRD codes as in Definition~\ref{D-Generic}, we need the following lemma.
A general version for arbitrary pairs $(\cF;\delta)$ can be derived as well, but is not needed for the rest of this paper.

\begin{lemma}\label{L-NormMat}
Consider $\cF=[m]\times[n]$ and $\delta=n$, thus $\ell=1$ and $N=\numin(\cF;\delta)=m$.
Recall the spaces $W_q$ and $\hat{W}_q$ from~\eqref{e-Wsets}.
Denote by $A_i^{(1)}$ the first column of the matrix~$A_i$ and define
\begin{align*}
  V_q&=\{(A_1,\ldots,A_m)\in (\F_q^{m\times n})^m\mid (A_1^{(1)},\ldots,A_m^{(1)})=I_m\}\subseteq W_q,\\
  \hat{V}_q&=\{(A_1,\ldots,A_m)\in V_q\mid \dd\subspace{A_1,\ldots,A_m}=n\}=\hat{W}_q\cap V_q.
\end{align*}
Then the proportion of $[m\times n;n]$-MRD codes in the space of all $m$-dimensional $(m\times n)$-rank-metric codes is given by
\begin{equation}\label{e-WV}
   \frac{|\hat{T}_q|}{|T_q|}=\frac{|\hat{W}_q|}{|W_q|}=\frac{|\hat{V}_q|}{|V_q|}\,\frac{\prod_{i=0}^{m-1}(q^{mn}-q^{i+m(n-1)})}{\prod_{i=0}^{m-1}(q^{mn}-q^i)},
\end{equation}
and thus $\lim_{q\rightarrow\infty}|\hat{T}_q|/|T_q|=\lim_{q\rightarrow\infty}|\hat{V}_q|/|V_q|$.
\end{lemma}

\begin{proof}
The stated identity for the limits is clear since the rightmost factor approaches~$1$ as $q\rightarrow\infty$.
The first identity in~\eqref{e-WV} is already in~\eqref{e-TW}, and thus we need to establish the second identity.
Reading each matrix~$A_i$ columnwise as a vector in $\F^{mn}$, we may identify $(\F_q^{m\times n})^m$ with
$\F_q^{mn\times m}$.
Then
\[
    W_q=\{M\in\F_q^{mn\times m}\mid \rk(M)=m\}\ \text{ and }\
    V_q=\{(I_m\mid B)\T\mid B\in\F_q^{m\times m(n-1)}\}.
\]
Notice also that, thanks to $\delta=n$, the first columns of any tuple $(A_1,\ldots,A_m)$ in $\hat{W}_q$ are linearly independent.
Thus $\hat{W}_q\subseteq \{(B_1\mid B_2)\T\mid B_1\in\GL_m(\F_q),\,B_2\in\F_q^{m\times m(n-1)}\}$.
This shows that $|\hat{W}_q|=|\hat{V}_q|\gamma_m(q)$, where $\gamma_m(q)=|\GL_m(\F_q)|$.
Furthermore, $|V_q|=q^{m^2(n-1)}$ and $|W_q|=\prod_{i=0}^{m-1}(q^{mn}-q^i)$.
Using that $\gamma_m(q)=\prod_{i=0}^{m-1}(q^m-q^i)$, we arrive at
\[
    \frac{|\hat{W}_q|}{|W_q|}=\frac{|\hat{V}_q|}{|V_q|}\!\cdot\!\frac{\gamma_m(q)q^{m^2(n-1)}}{\prod_{i=0}^{m-1}(q^{mn}-q^i)}
    =\frac{|\hat{V}_q|}{|V_q|}\!\cdot\!
    \frac{\prod_{i=0}^{m-1}(q^{mn}-q^{i+m(n-1)})}{\prod_{i=0}^{m-1}(q^{mn}-q^i)}.
    \qedhere
\]
\end{proof}

In the case where $\delta=n=2$, the probability determined in Proposition~\ref{P-SmallMRD}  is the fraction $|\hat{V}_q|/|V_q|$,
and thus~\eqref{e-WV} leads to the following proportion.

\begin{cor}\label{C-SmallMRD}
The proportion of $[m\times2;2]_q$-MRD codes within the space of all $m$-dimensional rank-metric codes in~$\F_q^{m\times2}$
is given by
\[
   \frac{s_m(q)}{q^{m^2}}\!\cdot\!\frac{\prod_{i=0}^{m-1}(q^{2m}-q^{i+m})}{\prod_{i=0}^{m-1}(q^{2m}-q^i)},
\]
and converges to $\sum_{j=0}^m\frac{(-1)^j}{j!}$ as $q\rightarrow\infty$.
\end{cor}

For more general cases we obtain more conditions for the rank distance.
Since these conditions are not independent events on the random entries, we can only provide upper bounds on the probability by restricting to a subset of independent events.

\begin{theo}\label{T-RandMRD}
Let $\F=\F_q,\ \delta\in[n]$, and $\ell=n-\delta+1$.
For $(\alpha,\beta)\in[m]\times[\ell]$ let $B_{\alpha,\beta}=\big(a^{(\alpha,\beta)}_{i,j}\big)\in\F^{m\times (n-\ell)}$ be randomly chosen matrices
and set
\[
   A_{\alpha,\beta}=\big(\underbrace{0\mid\cdots\mid0\mid e_\alpha\mid 0\mid\cdots\mid0}_{\ell \text{ columns}}\mid B_{\alpha,\beta}\big)\in\F^{m\times n},
\]
where $e_{\alpha}$, the $\alpha$-th standard basis vector in~$\F^m$, is in the $\beta$-th column.
Then the rank-metric code $\cC=\big\langle A_{\alpha,\beta}\mid (\alpha,\beta)\in[m]\times[\ell]\,\big\rangle$ satisfies
\[
  \Prob\big(\cC\text{ is an $[m\times n;\delta]$-MRD code}\big)
  \leq \Big(\frac{s_m(q)}{q^{m^2}}\Big)^{(\delta-1)\ell}.
\]
\end{theo}

\begin{proof}
Note that by construction the matrices $A_{1,1},\ldots,A_{m,\ell}$ are linearly independent and thus
$\dim\cC=m\ell=m(n-\delta+1)$, as desired.
Hence it remains to discuss the rank distance.
In order to do so, we consider, for all fixed~$\beta$,
linear combinations of the form $\sum_{\alpha=1}^m\lambda_\alpha A_{\alpha,\beta}$.
For $(\beta,j)\in[\ell]\times[\delta-1]$ define
\begin{equation}\label{-e-Mbetaj}
    M_{\beta,j}=\begin{pmatrix}a_{1,j}^{(1,\beta)}& \cdots &  a_{1,j}^{(m,\beta)}\\ \vdots& &\vdots\\
               a_{m,j}^{(1,\beta)}&\cdots&a_{m,j}^{(m,\beta)}\end{pmatrix}\in\F^{m\times m}.
\end{equation}
Thus $M_{\beta,j}$ consists of the $j$-th columns of $B_{1,\beta},\ldots,B_{m,\beta}$.
Let us now consider the linear combination $\sum_{\alpha=1}^m\lambda_\alpha A_{\alpha,\beta}$.
After deleting the $\ell-1$ zero columns, this matrix has the form
\begin{equation}\label{e-sumA}
   \big(\lambda\mid M_{\beta,1}\lambda\mid \ldots\mid M_{\beta,n-\ell}\lambda\big),\text{ where }
   \lambda=(\lambda_1,\ldots,\lambda_m)\T.
\end{equation}
As a consequence, $\rk\big(\sum_{\alpha=1}^m\lambda_\alpha A_{\alpha,\beta}\big)=\delta$
implies that $\lambda$ is not an eigenvector of any $M_{\beta,j}$.
Since this has to be true for all $\lambda\in\F^m\setminus 0$, we conclude that
$\sigma(M_{\beta,j})=\emptyset$ for all $j=1,\ldots,\delta-1$.
All of this shows that if~$\dd(\cC)=\delta$, then $\sigma(M_{\beta,j})=\emptyset$
for all $(\beta,j)\in[\ell]\times[\delta-1]$.
Since the~$\ell(\delta-1)$ matrices $M_{\beta,j}$ are independently chosen, the probability of the latter is $(s_m(q)/q^{m^2})^{(\delta-1)\ell}$, as desired.
\end{proof}

Note that in the above proof we ignore an abundance of further conditions on the data $a_{i,j}^{(\alpha,\beta)}$ and therefore
the probability is in fact much smaller than the given upper bound.
However, these additional conditions are not independent and therefore difficult to quantify.

Let us have a closer look at the case where $\delta=n$, thus $\ell=1$.
In this case the above proof tells us the following.

\begin{cor}\label{C-SpecMRD}
Consider the situation of Theorem~\ref{T-RandMRD} with $\delta=n$, thus $\ell=1$.
Then
\[
   \cC\text{ is an }[m\times n;n]_q\text{-MRD code}
  \Longleftrightarrow \sigma\Big(\sum_{j=1}^{n-1} \mu_{j}M_{1,j}\Big)=\emptyset\text{ for all }(\mu_1,\ldots,\mu_{n-1})\in \F_q^{n-1}\setminus0.
\]
\end{cor}

\begin{proof}
Since $\delta=n$, the code~$\cC$ is given by $\{\sum_{\alpha=1}^m\lambda_\alpha A_{\alpha,1}\mid \lambda_\alpha\in\F_q\}$, and for
any $\lambda=(\lambda_1,\ldots,\lambda_m)\in\F_q^m\setminus0$ the matrix $A(\lambda):=\sum_{\alpha=1}^m\lambda_{\alpha}A_{\alpha,1}$
equals the matrix in~\eqref{e-sumA}.
We thus obtain
$\rk A(\lambda)=n$ iff  $\mu_0\lambda\not=\sum_{j=1}^{n-1}\mu_j M_{1,j}\lambda$ for all
$(\mu_0,\ldots,\mu_{n-1})\in\F_q^{n}\setminus0$. This leads to the desired equivalence.
\end{proof}

\begin{exa}\label{E-MRDUpperB}
For $[4\times3;3]_q$-MRD codes we conducted computer experiments consisting of 10 million trials, each of which generated~$2$ random $4\times 4$ matrices over~$\F_q$,
serving as $M_{1,1}$ and $M_{1,2}$ in the proof of Corollary \ref{C-SpecMRD}.
In each trial, we checked if all nontrivial linear combinations of these two matrices were spectrum-free -- or equivalently, if the associated matrices $A_{1},\ldots,A_{4}\in\F_q^{4\times3}$ generated MRD codes.
The table in~\eqref{e-table} presents, for various values of~$q$, the estimated relative frequencies of spectrum-free subspaces $\subspace{M_{1,1},M_{1,2}}$.
In other words, this estimates the proportion $|\hat{V}_q|/|V_q|$ from Lemma~\ref{L-NormMat}.
Next,~\eqref{e-WV} tells us that multiplying these proportions by
$\prod_{i=0}^3(q^{12}-q^{i+8})/(q^{12}-q^i)$ gives us the proportion of MRD codes inside the space of all
$4$-dimensional rank-metric codes in $\F_q^{4\times3}$.
We also compare our findings with the upper bound given in Theorem~\ref{T-RandMRD}.
The frequency for $q=2$ was performed by exhaustive search, instead of by random experiment.
\begin{equation}\label{e-table}
          \begin{array}{c||c|c|c|c|c} q& 2&3&5&7&11\\\hline\hline
               \text{Upper Bound}&0.008&0.0313&0.065&0.083&0.102\\\hline
                |\hat{V}_q|/|V_q|&0.0005357&0.0000689&0.0001913&0.00028&0.0003732\\\hline
                \text{Proportion of MRD codes}&0.000165&0.000039&0.000146&0.000234&0.000336\end{array}
\end{equation}
\end{exa}

We wish to point out that our results do not preclude the existence of parameter sets $(m,n,\delta)$ for which the proportion of $[m\times n;\delta]_q$-MRD codes
approaches~$0$ as $q\rightarrow\infty$.
In such a case, the non-MRD codes would be generic (and the MRD codes would be sparse in the language of~\cite{ByRa18}).

We conclude this paper with, once again, the Ferrers diagram $\cF=[1,3,3,4]$ and $\delta=3$.

\begin{cor}\label{C-Prob1334}
Let $\F=\F_q$. Consider the $4\times4$-Ferrers diagram $\cF=[1,3,3,4]$ and let $\delta=3$.
Let
\[
   A_1=\begin{pmatrix}1\!&\!0\!&\!a_{13}^1\!&\!a_{14}^1\\0\!&\!0\!&\!a_{23}^1\!&\!a_{24}^1\\ 0\!&\!0\!&\!a_{33}^1\!&\!a_{34}^1\\0\!&\!0\!&\!0\!&\!a_{44}^1\end{pmatrix},\
   A_2=\begin{pmatrix}0\!&\!1\!&\!a_{13}^2\!&\!a_{14}^2\\0\!&\!0\!&\!a_{23}^2\!&\!a_{24}^2\\ 0\!&\!0\!&\!a_{33}^2\!&\!a_{34}^2\\0\!&\!0\!&\!0\!&\!a_{44}^2\end{pmatrix},\
   A_3=\begin{pmatrix}0\!&\!0\!&\!a_{13}^3\!&\!a_{14}^3\\0\!&\!1\!&\!a_{23}^3\!&\!a_{24}^3\\ 0\!&\!0\!&\!a_{33}^3\!&\!a_{34}^3\\0\!&\!0\!&\!0\!&\!a_{44}^3\end{pmatrix},\
   A_4=\begin{pmatrix}0\!&\!0\!&\!a_{13}^4\!&\!a_{14}^4\\0\!&\!0\!&\!a_{23}^4\!&\!a_{24}^4\\ 0\!&\!1\!&\!a_{33}^4\!&\!a_{34}^4\\0\!&\!0\!&\!0\!&\!a_{44}^4\end{pmatrix}
\]
be randomly chosen in $\F_q[\cF]$.
Then
\[
    \Prob(\subspace{A_1,\ldots,A_4}\text{ is a maximal $[\cF;3]_q$-code})\leq \frac{s_3(q)}{q^9}\prod_{i=2}^4\left(1-\frac{1}{q^i}\right)\frac{q^7-2q^4+q}{q^7}.
\]
The right hand side tends to~$1/3$ as $q\rightarrow\infty$.
\end{cor}

\begin{proof}
Consider a linear combination $\lambda_2A_2+\lambda_3A_3+\lambda_4A_4$.
If this matrix has rank~$3$ for all $(\lambda_2,\lambda_3,\lambda_4)\neq0$, then the submatrices
\[
   \begin{pmatrix}1&a_{13}^2\\0&a_{23}^2\\0&a_{33}^2\end{pmatrix},\
   \begin{pmatrix}0&a_{13}^3\\1&a_{23}^3\\0&a_{33}^3\end{pmatrix},\
   \begin{pmatrix}0&a_{13}^4\\0&a_{23}^4\\1&a_{33}^4\end{pmatrix}
\]
generate a $[3\times2;2]$-MRD code.
The probability for this is given by $s_3(q)/q^9$ according to Proposition~\ref{P-SmallMRD}.
Furthermore, the last columns of $A_2,A_3,A_4$ have to be linearly independent, and the according probability
is $q^{-12}\prod_{i=0}^2(q^4-q^i)=\prod_{i=2}^4(1-q^{-i})$.
Finally, the last two columns of~$A_1$ have to be linearly independent, which has a probability of
\[
  \frac{(q^3-1)\big((q-1)q^3+(q^3-q)\big)}{q^7}=\frac{q^7-2q^4+q}{q^7}.
\]
Since the events are independent, we obtain the stated upper bound.
\end{proof}

The probability $P_q:=\Prob(\subspace{A_1,\ldots,A_4}\text{ is a maximal $[\cF;3]_q$-code})$ can be related to the proportion
of maximal $[\cF;3]_q$-codes in the sense of Definition~\ref{D-Generic}.
Indeed, similarly to Lemma~\ref{L-NormMat}, one obtains
$|\hat{T}_q|/|T_q|=P_q\prod_{i=0}^3(q^{11}-q^{i+7})/(q^{11}-q^i)$.

\begin{exa}\label{E-F1334b}
Consider the scenario of the last corollary for $q=2$ and $q=3$.
Then the upper bound for the probability is given by
\begin{align*}
  &\Prob(\subspace{A_1,\ldots,A_4}\text{ is a maximal $[\cF;3]_2$-code})\leq 0.044,\\[.5ex]
  &\Prob(\subspace{A_1,\ldots,A_4}\text{ is a maximal $[\cF;3]_3$-code})\leq 0.1376.
\end{align*}
These estimates clearly leave out crucial conditions and therefore the true probabilities are much smaller.
Indeed, using SageMath and testing 1,000,000 quadruples of random matrices of the above form shows that the probability is about $0.00042$ for $q=2$ and about $0.0041$ for $q=3$.
For larger~$q$ the actual probability appears to be around $0.03$.
Yet, as we have seen in Example~\ref{E-F1334}, it is not hard to construct maximal $[\cF;3]_q$-codes over any field~$\F_q$.
\end{exa}

\section*{Open Problems}
We presented constructions of maximal $[\cF;\delta]_q$-codes for various classes of pairs $(\cF;\delta)$, but the general Conjecture~\ref{C-FConj} remains wide open.
The difficulty of the problem may in part be due to its highly  `noncanonical' nature in the sense that solutions, for most pairs $(\cF;\delta)$, depend on the choice of basis.
This is also evidenced by the genericity results of the last two sections leading to very different situations depending on the pair $(\cF;\delta)$.
While we do not entirely exclude the existence of a universal approach to the construction of maximal Ferrers diagram codes, we believe that further methods
tailored to specific types of pairs $(\cF;\delta)$ are necessary to settle the conjecture.
We list some specific questions that arise from our considerations.
\begin{alphalist}
\item Can one classify pairs $(\cF;\delta)$ according to the approachability of the construction problem?
         A first step would be the generalization of Corollary~\ref{C-Fbarn} to general rank~$\delta\geq2$, which would then
         tell us that maximal $[\cF;\delta]$-codes are generic if and only if $(\cF;\delta)$ is MDS-constructible.
\item The proofs of Theorems~\ref{T-Fn1} and~\ref{T-Fd2} leave some freedom in the choice of the basis~$B$.
         Can a suitable choice provide us with more specific maximal Ferrers diagram codes that can be exploited further, for instance, as in Example~\ref{E-F1334}?
\item Can one characterize the pairs $(\cF;\delta)$ for which maximal $[\cF;\delta]$-codes can be realized as subfield subcodes of $\F_{q^m}$-linear MRD codes with the same rank distance?
\item Can maximal $[\cF;\delta]_q$-codes be realized as subcodes of $\F_q$-linear MRD codes, for instance those presented in~\cite{CKWW16,Sh16,TrZh18}?
         The simplest case may be  $m=n=\delta$.
          In this case MRD codes are known as spreadsets in finite geometry and well studied.
\item Can the constructions in Section~\ref{S-Recurs} be generalized to other highly regular Ferrers shapes and other ranks?
\item Are there pairs $(\cF;\delta)$ for which the asymptotic proportion of maximal $[\cF;\delta]$-codes approaches~$0$?
Are there even parameter sets $(m,n,\delta)$ for which the asymptotic proportion of $[m\times n;\delta]$-MRD codes approaches~$0$?
\end{alphalist}

\appendix
\section*{Appendix}
\setcounter{section}{1}
\setcounter{theo}{0}
\newtheorem{propA}[theo]{Proposition}

\textit{Proof of Theorem~\ref{T-Fdsolid}.}
Let~$\F$ be a field with at least $\max\{|D_i\cap\cF|-1\mid i=\delta,\ldots,m\}$ elements.
We follow the construction in~\cite[Thm.~32]{GoRa17}.
Let $\cI:=\{i: |D_i\cap\cF|-\delta+1>0\}=\{i_1,\ldots,i_z\}$.
For all $i\in\cI$ set $n_i:=|D_i\cap\cF|$ and choose a matrix $G_i\in \F^{(n_i-\delta+1)\times n_i}$ such that every full size minor is nonzero.
For finite fields this simply means that $G_i$ is the generator matrix of an MDS code and thus exists due to our condition on the field size.
If~$\F$ is infinite such matrices also exist: consider the entries as distinct indeterminates over~$\F$.
Then the full-size
minors are distinct nonzero polynomials, and choosing a point outside the variety (over~$\F$) of these minors,
provides the entries of the desired matrix.
Now we have $\wtH(uG_i)\geq\delta$ for all $u\in \F^{n_i-\delta+1}\setminus0$, where $\wtH(v)=|\{j\mid v_j\neq0\}|$ denotes the Hamming weight,
just like for vectors over finite fields.
For $(v_{i_1},\ldots,v_{i_z})\in\rs(G_{i_1})\times\ldots\times\rs(G_{i_z})$ define $A:=A(v_{i_1},\ldots,v_{i_z})\in \F^{m\times n}$ as the matrix with
the vector $v_{i_j}$ at the positions of~$D_{i_j}\cap\cF$ (which has indeed cardinality~$n_{i_j}$) and set all other entries equal to zero.
Define
\[
   \cC=\{A(v_{i_1},\ldots,v_{i_z})\mid (v_{i_1},\ldots,v_{i_z})\in\rs(G_{i_1})\times\ldots\times\rs(G_{i_z})\}.
\]
By construction $\cC\subseteq \F[\cF]$ (note that we do not make use of dots of~$\cF$ outside the specified diagonals).
Furthermore, $\dim\cC=\sum_{i\in\cI} (n_i-\delta+1)=\sum_{i=\delta}^m\max\{|D_i\cap\cF|-\delta+1,\,0\}$.
Finally, $\dd(\cC)=\delta$, which can be seen as follows.
Choose any nonzero matrix $A\in\cC$.
Let~$t\in\cI$ be maximal such that the $t$-th diagonal of~$A$ is nonzero.
By construction this diagonal contains at least~$\delta$ nonzero entries and therefore $\rk(A)\geq\delta$.
The rest is obvious or follows from Theorem~\ref{T-Fbar}.
\hfill$\square$

\smallskip

In the rest of this appendix we show that the assumption used in~\cite[Thm.~7]{EGRW16} for the construction of maximal $[\cF;\delta]$-codes over sufficiently large fields
is equivalent to $(\cF;\delta)$ being MDS-constructible.
We need the following notions.
Fix any $\alpha\in\{0,\ldots,\delta-1\}$ such that $\numin(\cF;\delta)=\nu_\alpha(\cF;\delta)$.
Denote by~$\cF_{(\alpha)}$ the Ferrers diagram obtained by deleting the first~$\alpha$ rows and last $\delta-1-\alpha$ columns from~$\cF$.
Thus $\numin(\cF;\delta)=|\cF_{(\alpha)}|$.
We call the diagonal~$D_s$ an \emph{MDS diagonal} of $(\cF;\delta)$ w.r.t.~$\alpha$ if it satisfies:
\begin{alphalist}
\item $|D_s\cap (\cF\setminus\cF_{(\alpha)})|=\delta-1$. In other words,~$D_s$ has $\alpha$ dots in the
         first~$\alpha$ rows and $\delta-1-\alpha$ dots in the last $\delta-1-\alpha$ columns of~$\cF$.
\item There are no dots in~$\cF_{(\alpha)}$ below the diagonal $D_s$ and there is at least one dot  in~$\cF_{(\alpha)}$ on~$D_s$.
\end{alphalist}
It is shown in \cite[Thm.~7]{EGRW16} that if $(\cF;\delta)$ has an MDS diagonal, then maximal $[\cF;\delta]$-codes over sufficiently large fields can be
constructed with the aid of MDS codes (similarly to the construction in Theorem~\ref{T-Fdsolid}).
In fact we have

\begin{propA}\label{P-SFsolid}
Given any pair $(\cF;\delta)$. Then
\[
    (\cF;\delta) \text{ has an MDS diagonal}\Longleftrightarrow (\cF;\delta)\text{ is MDS-constructible}.
\]
\end{propA}

\begin{proof}
Consider Figure~\ref{F-Rectangle} in which we indicate the row indexed by~$\alpha$ and the column indexed by $n-\delta+2+\alpha$.
Thus the lower left corner contains the Ferrers diagram $\cF_{(\alpha)}$.
\begin{figure}[ht]
\centering
\begin{tikzpicture}[scale=0.42]
    \draw (0,0) to (11,0);
    \draw (0,0) to (0,13);
    \draw (11,0) to (11,13);
    \draw (0,13) to (11,13);
    \draw (0,8.5) to (7,8.5);
    \draw (7,0) to (7,8.5);
    \draw (1.9,13) to (11,3.9);
    \draw (11.5,3.8) node (b1) [label=center:${\scriptstyle D_{\delta}}$] {};
    \draw (-0.5,8.5) node (b1) [label=center:${\scriptstyle \alpha}$] {};
     \draw (6.7,8.2) node (b1) [label=center:$\bullet$] {};
    \draw (6.8,-0.5) node (b1) [label=center:${\scriptstyle n\!-\!\delta\!+\!2\!+\!\alpha}$] {};
\end{tikzpicture}
\caption{MDS diagonal vs. MDS-constructible}
\label{F-Rectangle}
\end{figure}
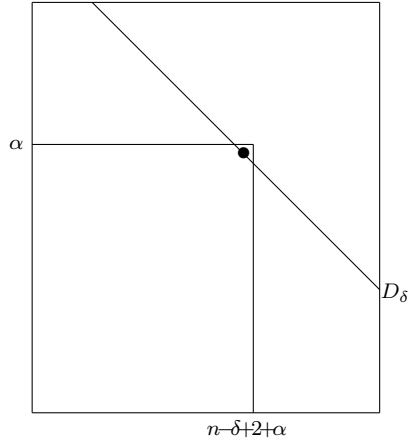
The upper right dot in~$\cF_{(\alpha)}$ is at position $(\alpha+1,\, n-\delta+1+\alpha)$ and thus on the diagonal $D_{\delta}$.
Therefore the only diagonals potentially intersecting with~$\cF_{(\alpha)}$ are $D_s$ where $s\geq\delta$.
These are also the only diagonals that may contribute to $\sum_{i=\delta}^m\max\{|D_i\cap\cF|-\delta+1,\,0\}$.
We compute
\begin{align*}
 (\cF;\delta)\text{ is MDS-constructible}&\Longleftrightarrow |\cF_{(\alpha)}|=\sum_{s=\delta}^m\max\{|D_s\cap\cF|-\delta+1,\,0\}\\
   &\mbox{}\hspace*{-7em}\Longleftrightarrow |D_s\cap\cF|=\delta-1+|D_s\cap\cF_{(\alpha)}|\text{ for all $s$ such that }D_s\cap\cF_{(\alpha)}\neq\emptyset\\
   &\mbox{}\hspace*{-7em}\Longleftrightarrow D_s\text{ has $\delta-1$ dots outside $\cF_{(\alpha)}$ for all $s$ such that }D_s\cap\cF_{(\alpha)}\neq\emptyset\\
   &\mbox{}\hspace*{-7em}\Longleftrightarrow D_{\tilde{s}}\text{ is an MDS diagonal, where~$\tilde{s}$ is maximal such that }D_{\tilde{s}}\cap\cF_{(\alpha)}\neq\emptyset.
   \qedhere
\end{align*}
\end{proof}

\bibliographystyle{abbrv}
\bibliography{literatureAK,literatureLZ}

\begin{thebibliography}{10}

\bibitem{Ant19}
J.~Antrobus.
\newblock {\em The State of Lexicodes and Ferrers Diagram Rank-Metric Codes}.
\newblock PhD thesis, University of Kentucky, 2019.
\newblock To appear in June 2019.

\bibitem{Bal15}
E.~Ballico.
\newblock Linear subspaces of matrices associated to a {F}errers diagram and
  with a prescribed lower bound for their rank.
\newblock {\em Lin. Algebra Appl.}, 483:30--39, 2015.

\bibitem{BEGR16}
E.~Ben-{S}asson, T.~Etzion, A.~Gabizon, and N.~Raviv.
\newblock Subspace polynomials and cyclic subspace codes.
\newblock {\em IEEE Trans. Inform. Theory}, IT-62:1157--1165, 2016.

\bibitem{ByRa18}
E.~Byrne and A.~Ravagnani.
\newblock Partition-balanced families of codes and asymptotic enumeration in
  coding theory.
\newblock Preprint 2018. arXiv:1805.02049, 2018.

\bibitem{CWJ03}
P.~A. Chou, Y.~Wu, and K.~Jain.
\newblock Practical network coding.
\newblock In {\em Proc.~2003 Allerton Conf. Communications, Control and
  Computing}, Monticello, IL, 2003.

\bibitem{CKWW16}
J.~{de la Cruz}, M.~Kiermeier, A.~Wassermann, and W.~Willems.
\newblock Algebraic structures of {MRD} codes.
\newblock {\em Adv. Math. Commun.}, 10:499--510, 2016.

\bibitem{Del78}
P.~Delsarte.
\newblock Bilinear forms over a finite field, with applications to coding
  theory.
\newblock {\em J.\ Combin. Theory Ser.\ A}, 25:226--241, 1978.

\bibitem{EGRW16}
T.~Etzion, E.~Gorla, A.~Ravagnani, and A.~Wachter-Zeh.
\newblock Optimal {F}errers diagram rank-metric codes.
\newblock {\em IEEE Trans. Inform. Theory}, IT-62:1616--1630, 2016.

\bibitem{EtSi09}
T.~Etzion and N.~Silberstein.
\newblock Error-correcting codes in projective spaces via rank-metric codes and
  {F}errers diagrams.
\newblock {\em IEEE Trans. Inform. Theory}, IT-55:2909--2919, 2009.

\bibitem{Gab85}
E.~M. Gabidulin.
\newblock Theory of codes with maximal rank distance.
\newblock {\em Probl. Inf. Transm.}, 21:1--12, 1985.

\bibitem{GaPi13}
E.~M. Gabidulin and N.~I. Pilipchuk.
\newblock Rank subcodes in multicomponent network coding.
\newblock {\em Probl. Inf. Trans. (Engl. Transl.)}, 49:40--53, 2013.

\bibitem{GPB10}
E.~M. Gabidulin, N.~I. Pilipchuk, and M.~Bossert.
\newblock Decoding of random network codes.
\newblock {\em Probl. Inf. Trans. (Engl. Transl.)}, 46:300--320, 2010.

\bibitem{GLMT15}
H.~Gluesing-Luerssen, K.~Morrison, and C.~Troha.
\newblock Cyclic orbit codes and stabilizer subfields.
\newblock {\em Adv. Math. Commun.}, 9:177--197, 2015.

\bibitem{GLT16}
H.~Gluesing-Luerssen and C.~Troha.
\newblock Construction of subspace codes through linkage.
\newblock {\em Adv. Math. Commun.}, 10:525--540, 2016.

\bibitem{GoRa14}
E.~Gorla and A.~Ravagnani.
\newblock Partial spreads in random network coding.
\newblock {\em Finite Fields Appl.}, 26:104--115, 2014.

\bibitem{GoRa17}
E.~Gorla and A.~Ravagnani.
\newblock Subspace codes from {F}errers diagrams.
\newblock {\em J. Algebra Appl.}, 16, 2017.
\newblock DOI:10.1142/S0219498817501316.

\bibitem{HeKu17}
D.~Heinlein and S.~Kurz.
\newblock Coset construction for subspace codes.
\newblock {\em IEEE Trans. Inform. Theory}, IT-63:7651--7660, 2017.

\bibitem{HKMKE03}
T.~Ho, R.~Koetter, M.~M{\'e}dard, D.~Karger, and M.~Effros.
\newblock The benefits of coding over routing in a randomized setting.
\newblock In {\em Proc.~2003 IEEE Int.\ Symp.\ Information Theory}, page 442,
  Yokohama/Japan, 2003.

\bibitem{KoKsch08}
R.~Koetter and F.~R. Kschischang.
\newblock Coding for errors and erasures in random network coding.
\newblock {\em IEEE Trans. Inform. Theory}, IT-54:3579--3591, 2008.

\bibitem{LCF19}
S.~Liu, Y.~Chang, and T.~Feng.
\newblock Constructions for optimal {F}errers diagram rank-metric codes.
\newblock {\em IEEE Trans. Inform. Theory}, 2019.
\newblock DOI 10.1109/TIT.2019.2894401.

\bibitem{MaVa13}
H.~Mahdavifar and A.~Vardy.
\newblock Algebraic list-decoding of subspace codes.
\newblock {\em IEEE Trans. Inform. Theory}, IT-59:7814--7828, 2013.

\bibitem{Mor06}
K.~E. Morrison.
\newblock Integer sequences and matrices over finite fields.
\newblock {\em J. Integer Seq.}, 9, 2006.
\newblock Article 06.2.1.

\bibitem{NHTRR18}
A.~Neri, A.-L. Horlemann-Trautmann, T.~Randrianarisoa, and J.~Rosenthal.
\newblock On the genericity of maximum rank distance and {G}abidulin codes.
\newblock {\em Des. Codes Cryptogr.}, 86:341--363, 2018.

\bibitem{PNLS17}
S.~Puchinger, J.~{Rosenkilde n{\'e} Nielsen}, W.~Li, and V.~Sidorenko.
\newblock Row reduction applied to decoding of rank metric and subspace codes.
\newblock {\em Des. Codes Cryptogr.}, 82:389--409, 2017.

\bibitem{Ro91}
R.~M. Roth.
\newblock Maximum-rank array codes and their application to crisscross error
  correction.
\newblock {\em IEEE Trans. Inform. Theory}, IT-37:328--336, 1991.

\bibitem{Schw79}
J.~T. Schwartz.
\newblock Probabilistic algorithms for verification of polynomial identities.
\newblock In {\em International Symposium on Symbolic and Algebraic
  Manipulation}, pages 200--215. Springer, 1979.

\bibitem{Sh16}
J.~Sheekey.
\newblock A new family of linear maximum rank distance codes.
\newblock {\em Adv. Math. Commun.}, 10:475--488, 2016.

\bibitem{SiTr15}
N.~Silberstein and A.-L. Trautmann.
\newblock Subspace codes based on graph matchings, {F}errers diagrams, and
  pending blocks.
\newblock {\em IEEE Trans. Inform. Theory}, IT-61:3937--3953, 2015.

\bibitem{SiKsch09p}
D.~Silva and F.~R. Kschischang.
\newblock Fast encoding and decoding of {G}abidulin codes.
\newblock In {\em 2009 IEEE International Symposium on Information Theory
  (ISIT)}, pages 2858--2862, 2009.

\bibitem{SKK08}
D.~Silva, F.~R. Kschischang, and R.~K{\"o}tter.
\newblock A rank-metric approach to error control in random network coding.
\newblock {\em IEEE Trans. Inform. Theory}, IT-54:3951--3967, 2008.

\bibitem{Sto88}
R.~Stong.
\newblock Some asymptotic results on finite vector spaces.
\newblock {\em Adv. Applied Math}, 9:167--199, 1988.

\bibitem{TMBR13}
A.-L. Trautmann, F.~Manganiello, M.~Braun, and J.~Rosenthal.
\newblock Cyclic orbit codes.
\newblock {\em IEEE Trans. Inform. Theory}, IT-59:7386--7404, 2013.

\bibitem{TrZh18}
R.~Trombetti and Y.~Zhou.
\newblock A new family of {MRD} codes in \protect{$\F_q^{2n\times2n}$} with
  right and middle nuclei \protect{$\F_{q^n}$}.
\newblock {\em IEEE Trans. Inform. Theory}, IT-65:1054--1062, 2018.

\bibitem{WAS13}
A.~Wachter-Zeh, V.~Afanassiev, and V.~Sidorenko.
\newblock Fast decoding of {G}abidulin codes.
\newblock {\em Des. Codes Cryptogr.}, 66:57--73, 2013.

\bibitem{ZhGe19}
T.~Zhang and G.~Ge.
\newblock Constructions of optimal {F}errers diagram rank metric codes.
\newblock {\em Des. Codes Cryptogr.}, 87(1):107--121, 2019.

\bibitem{Zi79}
R.~Zippel.
\newblock Probabilistic algorithms for sparse polynomials.
\newblock In {\em International Symposium on Symbolic and Algebraic
  Manipulation}, pages 216--226. Springer, 1979.

\end{thebibliography}
\end{document}